\documentclass[11pt]{amsart}
\usepackage[T1]{fontenc}
\usepackage[utf8]{inputenc}
\usepackage{newtxtext}
\usepackage{newtxmath}
\usepackage{microtype}
\usepackage{amsmath,amsthm,mathtools,mathrsfs,bm}
\usepackage{url}
\usepackage{graphicx}
\usepackage{booktabs}
\usepackage{array}
\usepackage{enumitem}
\usepackage{etoolbox}
\usepackage{placeins}
\usepackage{float}
\usepackage[numbers,sort&compress]{natbib}
\usepackage[dvipsnames]{xcolor}
\usepackage[margin=1.05in]{geometry}
\pdftrailerid{}
\definecolor{linkblue}{RGB}{0,72,150}
\definecolor{citegreen}{RGB}{0,110,70}
\definecolor{urlpurple}{RGB}{105,55,135}
\usepackage[
  colorlinks=true,
  linkcolor=linkblue,
  citecolor=citegreen,
  urlcolor=urlpurple,
  pdftitle={},
  pdfauthor={},
  pdfsubject={},
  pdfkeywords={},
  pdfcreator={},
  pdfproducer={}
]{hyperref}
\usepackage[nameinlink,noabbrev,capitalise]{cleveref}
\numberwithin{equation}{section}
\newtheorem{theorem}{Theorem}[section]
\newtheorem{lemma}[theorem]{Lemma}
\newtheorem{proposition}[theorem]{Proposition}
\newtheorem{corollary}[theorem]{Corollary}
\newtheorem{definition}[theorem]{Definition}
\newtheorem{assumption}[theorem]{Assumption}

\theoremstyle{remark}\newtheorem{remark}[theorem]{Remark}
\crefname{theorem}{Theorem}{Theorems}
\crefname{lemma}{Lemma}{Lemmas}
\crefname{proposition}{Proposition}{Propositions}
\crefname{corollary}{Corollary}{Corollaries}
\crefname{definition}{Definition}{Definitions}
\crefname{assumption}{Assumption}{Assumptions}
\crefname{problem}{Problem}{Problems}
\crefname{remark}{Remark}{Remarks}
\crefname{figure}{Figure}{Figures}
\crefname{table}{Table}{Tables}
\crefname{section}{Section}{Sections}
\crefname{equation}{Eq.}{Eqs.}
\DeclareMathOperator{\Tr}{Tr}

\DeclareMathOperator{\diag}{diag}
\newcommand{\bbZ}{\mathbb Z}

\newcommand{\cA}{\mathfrak A}
\newcommand{\cD}{\mathcal D}
\newcommand{\cL}{\mathscr L}
\newcommand{\cS}{\mathcal S}

\newcommand{\norm}[1]{\lVert #1\rVert}

\usepackage{orcidlink}
\title{Quantum Turing Patterns}

\author{Kazuki Ikeda\,\orcidlink{0000-0003-3821-2669}}
\address{Department of Physics, University of Massachusetts Boston, 100 Morrissey Boulevard, Boston, Massachusetts 02125, USA}
\address{Center for Nuclear Theory, Department of Physics and Astronomy, Stony Brook University, Stony Brook, New York 11794-3800, USA}
\email{Kazuki.Ikeda@umb.edu}

\begin{document}
\raggedbottom
\begin{abstract}
We construct quantum Turing patterns in Lindblad lattice dynamics and establish a rigorous theory of their nonlinear order and quantum fluctuations. For an explicit family of completely positive lattice generators with finite-range couplings, the first-moment equations undergo a supercritical instability at a nonzero wave number and admit analytic site- and bond-centered commensurate stripe branches.  These branches are locally asymptotically stable in their reflection-fixed period-cell spaces, and projected coherent states exhibit extensive Bragg order on every bounded time interval in the semiclassical limit.  We prove $O(N^{-1/2})$ convergence of microscopic covariances to a nonautonomous Gaussian Lyapunov flow, transferring strict partial-transpose uncertainty violations to sufficiently large $N$.  In the homogeneous Gaussian sector, a single dimensionless ratio controls both the Turing stability determinant and the logarithmic negativity of opposite momenta, relating wavelength selection directly to quantum entanglement.  Differential transport shifts the strongest opposite-momentum correlations from the infrared to the selected Turing scale.  Numerical continuation and two-dimensional simulations display stripe, spot, and labyrinth morphologies whose Fourier modes and fluctuation spectra concentrate at the same selected wave numbers.
\end{abstract}
\maketitle

\section{Introduction and main results}
Classical Turing patterns arise when differential transport destabilizes a homogeneous state at a nonzero wave number while the homogeneous mode remains stable.  The mechanism organizes stripes, spots, and labyrinths across reaction--diffusion theory and experiment \cite{Turing1952,CrossHohenberg1993,GiererMeinhardt1972,MainiPainterChau1997,KondoMiura2010,Krause2021}.  Its biological role is supported by pigment rearrangement, hair-follicle spacing, digit specification, palatal-ruga formation, and systematic searches of biochemical networks \cite{KondoAsai1995,KondoMiura2010,KondoWatanabeMiyazawa2021,YamaguchiYoshimotoKondo2007,NakamasuTakahashiKanbeKondo2009,SickReinkerTimmerSchlake2006,RaspopovicMarconRussoSharpe2014,EconomouEtAl2012,PaulAdetunjiHong2024}.  These developments sharpen a more microscopic question: can the same pattern-selection mechanism emerge from quantum Markov dynamics?

A \emph{quantum Turing pattern} is a Turing pattern generated by Lindblad lattice dynamics and detected by microscopic noncommuting observables. We then determine how the selected wavelength appears in its quantum fluctuation spectrum.

Quantum systems have displayed Turing-type symmetry breaking in several settings. Levitov et al.\ predicted a triangular modulation in a spatially uniform two-dimensional exciton system, while Ardizzone et al. observed and controlled transverse Turing patterns in a coherent polariton fluid \cite{LevitovSimonsButov2005,Ardizzone2013}.  In open quantum systems, Bandyopadhyay et al.\ considered two coupled quantum oscillators, Kato and Nakao studied two quantum activator--inhibitor units, and Comparato et al.\ analyzed mode competition in a finite bosonic GKSL chain \cite{BandyopadhyayKhatunBanerjee2021,KatoNakao2022,ComparatoGarganoLoFranco2026}.  Chia et al.\ addressed Lindblad quantization of planar polynomial flows \cite{ChiaMokKwekNoh2025}.  These studies consider quantum-fluid patterns, few-mode open-system instabilities, and general Lindblad quantization, but none derives the critical wave number, nonlinear spatial branch, and microscopic fluctuation dynamics within a single lattice construction.

To our knowledge, this is the first rigorous constructive theory of quantum Turing patterns in a Lindblad lattice.  Using only finite-range couplings, the construction fixes a nonzero critical wave number and proves the associated nonlinear branch and volume-uniform Bragg order.  It also establishes covariance convergence for the microscopic dynamics and derives a closed relation between the Turing stability determinant and opposite-momentum entanglement.  The explicit family yields a lattice-commensurate stripe; an isotropic specialization produces two-dimensional spot and labyrinth morphologies whose fluctuation spectra concentrate on the selected momentum shell.

\begin{theorem}[Existence, stability, and Bragg order for quantum Turing patterns]\label[theorem]{thm:quantum-turing-patterns}
There exists an explicit family of translation-invariant Lindblad generators with finite-range couplings on two-dimensional bosonic lattices with the following properties.
\begin{enumerate}[label=\textup{(\roman*)}]
\item The first-moment limit undergoes a supercritical Turing bifurcation at a nonzero lattice wave number and admits analytic site- and bond-centered commensurate stripe branches.
\item The two branches are locally asymptotically stable in their respective reflection-fixed period-cell spaces.
\item For projected coherent product states prepared on either branch, the quantum structure factor has an extensive Bragg peak on every bounded time interval in the semiclassical limit.
\end{enumerate}
\end{theorem}
We denote the selected stripe wave number by $k_*$.  Section~3 gives the explicit coefficients, selected wave number, and branch asymptotics.

Beyond the explicit family, we prove two general results. Fix a lattice torus, a finite horizon $T<\infty$, and a bounded first-moment trajectory. Let $V_N(t)$ denote the microscopic covariance and $V(t)$ its Gaussian Lyapunov limit. Then
\[
 \sup_{0\le t\le T}\norm{V_N(t)-V(t)}=O(N^{-1/2}).
\]
For a homogeneous Gaussian momentum pair, let $\kappa_k$ denote its effective damping, $g_k$ its squeezing amplitude, and $\Omega_\lambda$ its rotation frequency, and write
\[
 R_k:=\sqrt{\kappa_k^2+4\Omega_\lambda^2},\qquad
 \eta_k:=\frac{2|g_k|}{R_k}.
\]
The pair is stable exactly when $\eta_k<1$; in that regime,
\[
 \widetilde\nu_-(k,-k)=\frac{1}{2(1+\eta_k)},\qquad
 E_{\rm LN}(k,-k)=\log_2(1+\eta_k).
\]
Here $\widetilde\nu_-$ is the smaller partially transposed symplectic eigenvalue and $E_{\rm LN}$ is the logarithmic negativity.

The accompanying Lean~4 formalization verifies the algebraic identities entering the spectral design, Lindblad parameter map, and homogeneous formulas; the source is included in the author's GitHub repository \url{https://github.com/IKEDAKAZUKI/Quantum-Turing-Pattern}.

\begin{figure}[!t]
\centering
\includegraphics[width=0.96\textwidth,height=0.66\textheight,keepaspectratio]{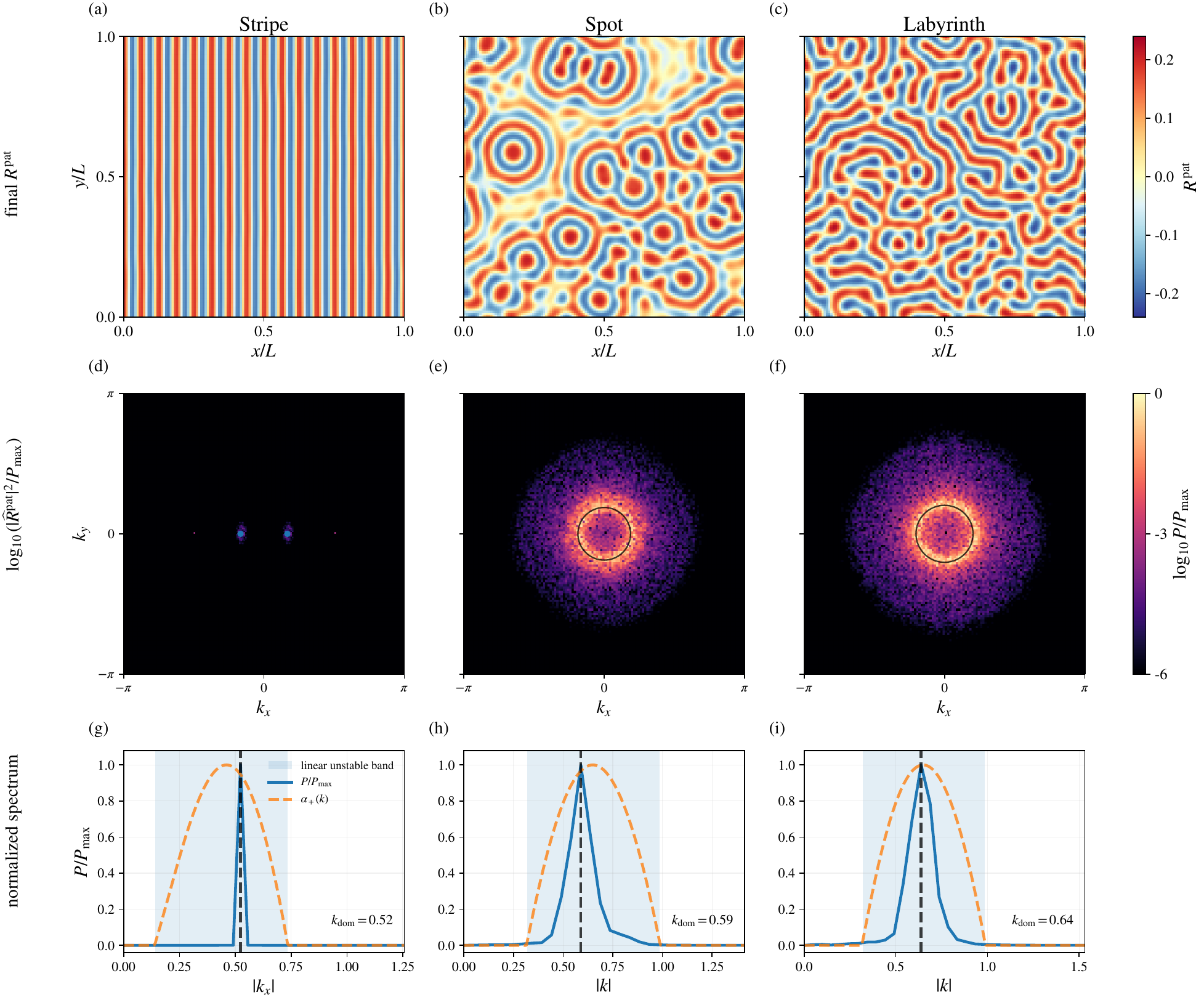}
\caption[Quantum Turing morphologies and selected wave numbers]{Quantum Turing morphologies and selected wave numbers in one Lindblad family.  The top row shows the first-moment fields, the middle row their two-dimensional Fourier power, and the bottom row the selected spectra together with the linear unstable bands.  The stripe selects the pair $\pm k_*$, while the isotropic spot and labyrinth regimes concentrate power on the selected wave-number shell.}
\label{fig:morph}
\end{figure}

\section{Quantum Markov dynamics and semiclassical estimates}
\subsection{Truncated dynamics and coherent states}
Let $d,L,N\in\mathbb N$ with $d,L,N\ge1$, let $e_j$ be the $j$th coordinate vector of $\bbZ^d$, and set $\Lambda_L=(\bbZ/L\bbZ)^d$.  The nearest-neighbor graph distance on this periodic lattice is denoted by $d_L$, and
\[
  B_r(x):=\{y\in\Lambda_L:d_L(x,y)\le r\}
\]
is its closed graph ball.  Let $a$ be the canonical annihilation operator on $\ell^2(\mathbb N)$, let $n=a^*a$, and let
\[
  \mathbf P_{\le M}:=\mathbf1_{\{n\le M\}}
\]
be the spectral projection of $n$ onto occupations $0,\ldots,M$.  More generally, for a self-adjoint operator $A$ and a Borel set $E$, $\mathbf1_E(A)$ denotes its spectral projection; expressions such as $\mathbf1_{\{n_x\ge m\}}$ use this functional calculus.\footnote{Throughout, $A^*$ denotes the Hilbert-space adjoint; hence $a_x^*$ is the creation operator and corresponds to $a_x^\dagger$ in common physics notation.  For superoperators, ${}^*$ denotes the trace-dual (Heisenberg) adjoint.}  For the Fock cutoff $M_N=\lfloor cN\rfloor$, $c>0$, set
\[
  \mathcal H_N=\mathbf P_{\le M_N}\ell^2(\mathbb N),\qquad
  a_{M_N}=\mathbf P_{\le M_N}a\mathbf P_{\le M_N},
\]
\[
  \mathcal H_{L,N}=\bigotimes_{x\in\Lambda_L}\mathcal H_N,
  \qquad
  \cA_{L,N}=B(\mathcal H_{L,N}).
\]
The copies at $x\in\Lambda_L$ are denoted by $a_{x,M_N}$ and
$n_{x,M_N}=a_{x,M_N}^*a_{x,M_N}$; when the cutoff is fixed we abbreviate $n_x:=n_{x,M_N}$.  For a density matrix $\rho$ and an observable $A$, write
\[
  \rho(A):=\Tr(\rho A),\qquad
  \rho_{\mathbf m,\mathbf n}:=\langle\mathbf m|\rho|\mathbf n\rangle,
  \qquad [s]_+:=\max\{s,0\},
\]
where $\mathbf m,\mathbf n\in\{0,\ldots,M_N\}^{\Lambda_L}$ are cutoff occupation multi-indices.
We use standard notation for quantum Markov semigroups and Lindblad generators \cite{Lindblad1976,Davies1976,BreuerPetruccione2002}.  For density matrices we use the convention
\[
  \cD[L]\rho=L\rho L^*-\frac12\{L^*L,\rho\},
\]
whose Heisenberg adjoint is $\cD[L]^*(A)=L^*AL-\frac12\{L^*L,A\}$.  The Heisenberg generator has Lindblad form
\begin{equation}\label{eq:lindblad}
\cL_{L,N}(A)=i[H_{L,N},A]+\sum_\varkappa\left(V_\varkappa^*AV_\varkappa-\frac12\{V_\varkappa^*V_\varkappa,A\}\right).
\end{equation}
For an initial state $\rho_{L,N}(0)$, its Schr\"odinger evolution and the corresponding expectation shorthand are
\[
  \rho_{L,N}(t):=e^{t\cL_{L,N}^*}\rho_{L,N}(0),
  \qquad \rho_t(A):=\rho_{L,N}(t)(A),
  \qquad \rho_0(A):=\rho_{L,N}(0)(A).
\]
For every lattice direction, write
\[
  \Delta_j f_x=f_{x+e_j}+f_{x-e_j}-2f_x,
  \qquad j=1,\ldots,d.
\]
Define the one-site rescaled quadratures and their copies by
\[
Q_N=\frac{a_{M_N}+a_{M_N}^*}{\sqrt{2N}},\qquad
P_N=\frac{a_{M_N}-a_{M_N}^*}{i\sqrt{2N}},
\qquad Q_{x,N},P_{x,N}\text{ at site }x.
\]
Since
\begin{equation}\label{eq:cutoff}
[a_{M_N},a_{M_N}^*]=I-(M_N+1)|M_N\rangle\langle M_N|,
\end{equation}
operator norm is not the correct topology for the oscillator cutoff.  We use a local coherent-state seminorm.

\begin{definition}\label[definition]{def:local-coherent-seminorm}
For $E<c$, let
\[
\psi_{\boldsymbol\alpha}^{(L,N)}=\bigotimes_{x\in\Lambda_L}
\frac{\mathbf P_{\le M_N}|\sqrt N\alpha_x\rangle}{\norm{\mathbf P_{\le M_N}|\sqrt N\alpha_x\rangle}},
\qquad \sup_x |\alpha_x|^2\le E.
\]
Let $\cS_{L,N}(E)$ be the convex hull of the rank-one states $|\psi_{\boldsymbol\alpha}^{(L,N)}\rangle\langle\psi_{\boldsymbol\alpha}^{(L,N)}|$.  For $A\in\cA_{L,N}$ set
\[
\norm{A}_{\cS_{L,N}(E)}=\sup_{\rho\in\cS_{L,N}(E)}\rho(A^*A)^{1/2}.
\]
\end{definition}

\begin{lemma}\label[lemma]{lem:comm}
For every $E<c$ there are constants $C,c_0>0$, independent of $L$, such that
\[
\rho(|M_N\rangle_x\langle M_N|)\le C e^{-c_0N},\qquad \rho\in\cS_{L,N}(E),\quad x\in\Lambda_L.
\]
\begin{equation}\label{eq:comm}
\norm{N[Q_{x,N},P_{x,N}]-iI}_{\cS_{L,N}(E)}\le CNe^{-c_0N}\to0.
\end{equation}
\end{lemma}
\begin{proof}
The occupation number in $|\sqrt N\alpha\rangle$ is Poisson with mean at most $NE$.  Since $M_N=\lfloor cN\rfloor$ and $c>E$, the Chernoff bound gives an exponentially small cutoff tail.  Equation \eqref{eq:comm} follows from \eqref{eq:cutoff}.  The product-state estimate is uniform in $L$ because the seminorm is defined through local coherent-state moments and the commutator is local.

\end{proof}

\begin{lemma}\label[lemma]{lem:symbol}
Fix a polynomial degree and a coherent-state energy bound $E<c$.  For every local polynomial $\mathscr P$ in the rescaled variables $a_{x,M_N}/\sqrt N$ and $a_{x,M_N}^*/\sqrt N$, normal ordering changes its coherent symbol by $O_{\cS_{L,N}(E)}(N^{-1})$, uniformly in $L$ and in the site labels appearing in $\mathscr P$.  The coherent first-moment equations obtained from \eqref{eq:lindblad} are classical polynomial symbol equations with $O_{\cS_{L,N}(E)}(N^{-1})$ remainders.
\end{lemma}
\begin{proof}
On the coherent-state set of \Cref{def:local-coherent-seminorm}, the cutoff tail is exponentially small by \Cref{lem:comm}.  Away from that tail the rescaled canonical commutator is $N^{-1}$ times the identity up to an exponentially small error.  Moving creation operators past annihilation operators in a fixed-degree polynomial produces only finitely many commutator terms, each smaller by one power of $N$.
\end{proof}

\subsection{Semiclassical estimates}
Throughout the uniform semiclassical estimates below, $\cL_{L,N}$ is a finite-range Lindblad generator with a fixed interaction range $r_0\in\mathbb N$.  Its coefficients are independent of $L$, $N$, and the cutoff boundary, and all scalar or matrix coefficients are uniformly bounded.  Its terms are: (i) on-site quadratic Hamiltonians, with optional coherent drives $\sqrt N(f_xa_x^*+\overline f_xa_x)$ satisfying $\sup_{L,N,x}|f_x|<\infty$; (ii) directional quadratic bond Hamiltonians; (iii) one-photon loss with rate $\kappa\ge0$ and directional bond jumps proportional to $a_x-a_y$; and (iv) on-site two-photon loss $(\gamma/N)\cD[a_x^2]$ with a fixed coefficient $\gamma>0$.  Every local term is supported in a set contained in some graph ball $B_{r_0}(x)$.  The direction-dependent realization constructed in \Cref{thm:lindblad}, for which $\gamma=2\nu>0$, belongs to this class.

\begin{proposition}\label[proposition]{prop:weighted-exp}
Let $r_0$ be the interaction range in the standing semiclassical assumptions, fix $\zeta>0$, and set
\begin{equation}\label{eq:exact-number-tilt}
 w_x(u)=e^{-\zeta d_L(x,u)},\qquad
 \mathcal N_x^{(\zeta)}=\sum_{u\in\Lambda_L}w_x(u)n_u,
 \qquad \Phi_{x,\theta}=e^{\theta\mathcal N_x^{(\zeta)}}.
\end{equation}
For sufficiently small fixed $\theta>0$ there is $C(\theta,\zeta)<\infty$, independent of $L,N$, the cutoff ratio $c$, and $M_N=\lfloor cN\rfloor$, such that every density matrix $\rho$ on the truncated space satisfies
\begin{equation}\label{eq:weighted-exponential-generator}
 \rho\!\left(\cL_{L,N}(\Phi_{x,\theta})\right)
 \le C(\theta,\zeta)N\rho(\Phi_{x,\theta}).
\end{equation}
The occupation-basis argument also yields a nearby-center form with a finite sum over $d_L(x,z)\le r_0$.
\end{proposition}
\begin{proof}
The truncated number operator $n_u=a_{u,M_N}^*a_{u,M_N}$ obeys, including on the top cutoff state,
\begin{equation}\label{eq:truncated-number-commutators}
 [n_u,a_{u,M_N}]=-a_{u,M_N},\qquad
 [n_u,a_{u,M_N}^*]=a_{u,M_N}^*.
\end{equation}
\begin{equation}\label{eq:exact-number-conjugation}
 e^{sn_u}a_{u,M_N}e^{-sn_u}=e^{-s}a_{u,M_N},\qquad
 e^{sn_u}a_{u,M_N}^*e^{-sn_u}=e^sa_{u,M_N}^*,
\end{equation}
with constants independent of the cutoff.  This is the point at which the weighted-number exponential observable differs from a plateau weight.

Write $\cL_{L,N}=\sum_Z\cL_Z$, where $\operatorname{diam}Z\le r_0$.  Bounded geometry and exponential summability give
\[
 \sup_{L,x}\sum_Z\sum_{u\in Z}w_x(u)<\infty,
\]
and weights within one interaction set are comparable by a factor depending only on $r_0$ and $\zeta$.

\medskip
\noindent\emph{Weighted monomial bounds.}
For a normal-ordered monomial $X_Z=\prod_{u\in Z}(a_u^*)^{p_u}a_u^{q_u}$ define
$\Delta_x(X_Z)=\sum_{u\in Z}(p_u-q_u)w_x(u)$.  Equation \eqref{eq:exact-number-conjugation} yields
\[
 \Phi_{x,\theta}X_Z\Phi_{x,\theta}^{-1}
 =e^{\theta\Delta_x(X_Z)}X_Z,
 \qquad
 |e^{\theta\Delta_x(X_Z)}-1|
 \le C_{r_0,\zeta}\theta\sum_{u\in Z}w_x(u)
\]
for every monomial occurring in the generator.  Reindexing occupation-basis matrix elements and using positivity in the form
$|\rho_{\mathbf m,\mathbf n}|\le(\rho_{\mathbf m,\mathbf m}\rho_{\mathbf n,\mathbf n})^{1/2}$ yields the following explicit bounds, where $u\sim v$ means that $\{u,v\}$ is the support of one finite-range two-site bond term of the generator and all constants are independent of the cutoff:
\begin{align}
 \bigl|\rho(i[a_u^2+a_u^{*2},\Phi_{x,\theta}])\bigr|
 &\le C\theta w_x(u)\rho((n_u+1)\Phi_{x,\theta}),\label{eq:weighted-onsite-squeezing}\\
 \bigl|\rho(i[a_u^*a_v+a_v^*a_u,\Phi_{x,\theta}])\bigr|
 &\le C\theta(w_x(u)+w_x(v))
 \rho((n_u+n_v+2)\Phi_{x,\theta}),\label{eq:weighted-hopping}\\
 \bigl|\rho(i[a_u^*a_v^*+a_ua_v,\Phi_{x,\theta}])\bigr|
 &\le C\theta(w_x(u)+w_x(v))
 \rho((n_u+n_v+2)\Phi_{x,\theta}),\label{eq:weighted-pairing}\\
 \bigl[\rho(\cD[a_u-a_v]^*(\Phi_{x,\theta}))\bigr]_+
 &\le C\theta(w_x(u)+w_x(v))
 \rho((n_u+n_v+2)\Phi_{x,\theta}),\label{eq:weighted-bond-dissipator}\\
 \bigl|\rho(i[\sqrt N(f_ua_u^*+\overline f_u a_u),\Phi_{x,\theta}])\bigr|
 &\le C\theta w_x(u)\rho(\sqrt{N(n_u+1)}\,\Phi_{x,\theta}).\label{eq:weighted-drive}
\end{align}
For the pairing term, the occupation-basis reindexing is as follows.  Let $\mathbf e_u$ denote one quantum at site $u$, write
\[
 \Phi_{x,\theta}|\mathbf n\rangle=\phi_{\mathbf n}|\mathbf n\rangle,
 \qquad
 \phi_{\mathbf n+\mathbf e_u+\mathbf e_v}
 =e^{\theta\Delta_{uv}}\phi_{\mathbf n},
 \qquad
 \Delta_{uv}=w_x(u)+w_x(v),
\]
and set $X=a_u^*a_v^*$.  The sum below is over occupations for which the shifted basis vector remains inside the cutoff.  Positivity of the density matrix and $2\sqrt{ab}\le a+b$ give
\begin{align*}
 \bigl|\rho(i[X+X^*,\Phi_{x,\theta}])\bigr|
 &\le
 2\bigl(e^{\theta\Delta_{uv}}-1\bigr)
 \sum_{\mathbf n}'\sqrt{(n_u+1)(n_v+1)}\,
 \phi_{\mathbf n}
 \bigl|\rho_{\mathbf n,\mathbf n+\mathbf e_u+\mathbf e_v}\bigr|\\
 &\le
 \bigl(e^{\theta\Delta_{uv}}-1\bigr)
 \sum_{\mathbf n}'\sqrt{(n_u+1)(n_v+1)}\,
 \phi_{\mathbf n}
 \left(\rho_{\mathbf n,\mathbf n}
 +\rho_{\mathbf n+\mathbf e_u+\mathbf e_v,
        \mathbf n+\mathbf e_u+\mathbf e_v}\right).
\end{align*}
In the second diagonal sum put
$\mathbf m=\mathbf n+\mathbf e_u+\mathbf e_v$; then
$\phi_{\mathbf m-\mathbf e_u-\mathbf e_v}
=e^{-\theta\Delta_{uv}}\phi_{\mathbf m}$ and
$2\sqrt{m_um_v}\le m_u+m_v$.  Together with
$2\sqrt{(n_u+1)(n_v+1)}\le n_u+n_v+2$ and
$e^{\theta\Delta_{uv}}-1\le C\theta(w_x(u)+w_x(v))$, this is \eqref{eq:weighted-pairing}.  The hopping estimate follows from an analogous one-quantum reindexing.  Expanding the bond dissipator into its two diagonal loss terms and two cross terms gives \eqref{eq:weighted-bond-dissipator}; the diagonal one-photon contributions are nonpositive.  The drive estimate follows from Cauchy--Schwarz.  Let $\cL_{L,N}^{\mathrm{rest}}$ denote the sum of the on-site quadratic-Hamiltonian, quadratic bond-Hamiltonian, bond-dissipator, and coherent-drive contributions to $\cL_{L,N}$.  Summing \eqref{eq:weighted-onsite-squeezing}--\eqref{eq:weighted-drive} over their supports gives
\begin{equation}\label{eq:positive-weighted-generator-bound}
 \left[
 \rho\!\left(\cL_{L,N}^{\mathrm{rest}}(\Phi_{x,\theta})\right)
 \right]_+
 \le C\theta\sum_u w_x(u)
 \rho\!\left([n_u+1+\sqrt{N(n_u+1)}]\Phi_{x,\theta}\right).
\end{equation}
The on-site one-photon-loss contribution is nonpositive.  The two-photon-loss contribution has the exact diagonal identity
\begin{equation}\label{eq:two-photon-exact-tilt}
 \frac\gamma N\cD[a_u^2]^*(\Phi_{x,\theta})
 =\frac\gamma N n_u(n_u-1)
 \bigl(e^{-2\theta w_x(u)}-1\bigr)\Phi_{x,\theta}.
\end{equation}
After decreasing $\theta$ if necessary,
$1-e^{-2\theta w}\ge c_0\theta w$ for $0\le w\le1$, and hence
\[
 \rho\!\left(\frac\gamma N\cD[a_u^2]^*(\Phi_{x,\theta})\right)
 \le-\frac{c_0\gamma\theta}{N}w_x(u)
 \rho(n_u(n_u-1)\Phi_{x,\theta}).
\]
The scalar estimate
\begin{equation}\label{eq:scalar-two-photon-absorption}
 C(n+1)+C\sqrt{N(n+1)}-\frac{c_0\gamma}{N}n(n-1)
 \le C'N,\qquad n\ge0,
\end{equation}
absorbs every high-occupation factor in \eqref{eq:positive-weighted-generator-bound} at its own site.  Summing over finite-range terms and using the uniform exponential-weight sum proves \eqref{eq:weighted-exponential-generator}.
\end{proof}

\subsubsection{Finite-time cutoff control}
A \emph{volume-uniform compact amplitude class} is a family $\mathcal K=\{\mathcal K_L\}_{L\ge1}$ in which each $\mathcal K_L\subset\ell^\infty(\Lambda_L;\mathbb C)$ is compact in the sup norm and
\[
  \sup_{L\ge1}\sup_{\boldsymbol\alpha\in\mathcal K_L}\sup_{x\in\Lambda_L}|\alpha_x|<\infty.
\]
All occurrences of a compact set $\mathcal K$ below have this meaning; for a fixed period cell it is obtained by periodic repetition on compatible tori.

Static coherent-state estimates require only $c>E$.  For propagation on a prescribed time interval, the cutoff ratio is chosen after the horizon:
\[
  \forall T<\infty\quad \exists c_T<\infty\quad \forall c\ge c_T.
\]

\begin{lemma}\label[lemma]{lem:cutoff-tail-propagation}
Fix $T<\infty$ and a volume-uniform compact amplitude class $\mathcal K$.  Let $E_0$ bound the initial occupations $|\alpha_x(0)|^2$ uniformly in volume.  There exists $c_T=c_T(\mathcal K)<\infty$ such that, for every cutoff ratio $c\ge c_T$ and $M_N=\lfloor cN\rfloor$, projected coherent product initial data with amplitudes in $\mathcal K$ satisfy the following estimate.  For every fixed buffer $C_{\rm buf}$ there are $\zeta,\theta,c_{\rm tail}>0$ and $C_{T,C_{\rm buf}}<\infty$, independent of $L,N$, and every larger $c$, such that
\begin{equation}\label{eq:cutoff-tail-propagation}
 \sup_{0\le t\le T}\sup_{x\in\Lambda_L}
 \rho_{L,N}(t)\!\left(\mathbf1_{\{n_x\ge M_N-C_{\rm buf}\}}\right)
 \le C_{T,C_{\rm buf}}e^{-c_{\rm tail}N}.
\end{equation}
In particular this probability is $O(N^{-K_0})$ for every fixed $K_0$.
\end{lemma}
\begin{proof}
Use the weighted-number exponential observable \eqref{eq:exact-number-tilt}.  At one site, the projected coherent occupation law is a Poisson random variable $X$ of mean $N|\alpha|^2$ conditioned on $X\le M_N$.  Conditioning on this lower event stochastically decreases $X$, so for every $t\ge0$,
\[
 \mathbb E[e^{tX}\mid X\le M_N]\le\mathbb E[e^{tX}].
\]
The product structure yields the cutoff-independent initial estimate
\begin{equation}\label{eq:initial-exact-number-mgf}
 \rho_0(\Phi_{x,\theta})
 \le\exp\!\left\{NE_0\sum_{y\in\Lambda_L}
 \bigl(e^{\theta e^{-\zeta d_L(x,y)}}-1\bigr)\right\}
 \le e^{Nm_0(\theta,\zeta)},
\end{equation}
where $m_0$ is uniform in $L$.  \Cref{prop:weighted-exp} and Gronwall imply
\[
 \sup_{0\le t\le T}\sup_x\rho_t(\Phi_{x,\theta})
 \le e^{Nm_T(\theta,\zeta)},\qquad
 m_T=m_0+TC(\theta,\zeta).
\]
On the boundary event, $\Phi_{x,\theta}\ge e^{\theta(M_N-C_{\rm buf})}$.  Markov's inequality and $M_N=\lfloor cN\rfloor$ give
\begin{equation}\label{eq:strong-cutoff-boundary-tail}
 \rho_t(n_x\ge M_N-C_{\rm buf})
 \le C_{C_{\rm buf},\theta}
 \exp\{-N[\theta c-m_T(\theta,\zeta)]\}.
\end{equation}
Fix $\zeta$ and then a sufficiently small $\theta$.  Choose $c_T$ so that
\[
 \delta_T:=\theta c_T-m_T(\theta,\zeta)>0,
 \qquad c_{\rm tail}:=\frac{\delta_T}{2}.
\]
For every $c\ge c_T$, the exponent in \eqref{eq:strong-cutoff-boundary-tail} is at least $2c_{\rm tail}N$; the floor and the fixed buffer are absorbed into the prefactor.  Hence \eqref{eq:cutoff-tail-propagation} holds with quantifier order
$\forall T\,\exists c_T\,\forall c\ge c_T$.
\end{proof}

\subsubsection{Centered moments and first-moment convergence}
For the evolved state define, before invoking any cubic closure,
\begin{equation}\label{eq:local-centered-fields}
\begin{aligned}
 \delta Q_{x,N}(t)&=Q_{x,N}-\rho_t(Q_{x,N}),\qquad
 \delta P_{x,N}(t)=P_{x,N}-\rho_t(P_{x,N}),\\
 \mathcal C_x(t)&=\rho_t\!\left(\delta Q_{x,N}(t)^2+\delta P_{x,N}(t)^2\right).
\end{aligned}
\end{equation}

\begin{lemma}\label[lemma]{lem:cubic-moment-closure}
Under the hypotheses of \Cref{lem:cutoff-tail-propagation}, fix a cutoff ratio $c\ge c_T$.  Let $\mathscr P_x$ be a local polynomial of total degree at most three in the rescaled quadratures $Q_{y,N},P_{y,N}$ with $y\in B_{r_0}(x)$.  The expression $\mathscr P_x(\rho_t(Q),\rho_t(P))$ means that every $Q_{y,N}$ and $P_{y,N}$ in the polynomial is replaced by its expectation $\rho_t(Q_{y,N})$ and $\rho_t(P_{y,N})$, respectively.  If the polynomial coefficients and first moments remain bounded on $[0,T]$, then
\begin{equation}\label{eq:cubic-moment-closure}
 \left|\rho_t(\mathscr P_x)-\mathscr P_x(\rho_t(Q),\rho_t(P))\right|
 \le C_{\mathscr P,T,c}\!\left(\max_{y\in B_{r_0}(x)}\mathcal C_y(t)+N^{-1}\right)
 +C_{\mathscr P,T,c}e^{-c_{\rm bd}N},
\end{equation}
uniformly in $L,N,x$, and $0\le t\le T$, for some $c_{\rm bd}>0$.
\end{lemma}
\begin{proof}
Normal order $\mathscr P_x$ and write each rescaled quadrature as its mean plus the centered operator in \eqref{eq:local-centered-fields}.  Terms with one centered factor vanish.  Each ordering change contributes one rescaled canonical commutator of size $O(N^{-1})$, apart from a boundary projector controlled by \Cref{lem:cutoff-tail-propagation}.

We make the noncommuting bulk insertion explicit.  Write a normal-ordered monomial as
\[
 M=\prod_{y\in B_{r_0}(x)}M_y,
 \qquad \deg M_y=d_y,
 \qquad d_M:=\sum_{y\in B_{r_0}(x)}d_y\le3,
\]
where factors at distinct sites commute, and set $\Pi_{y,m}=\mathbf1_{\{n_y\le m\}}$.  Repeated use of
$a_y\Pi_{y,m}=\Pi_{y,m-1}a_y\Pi_{y,m}$ and
$a_y^*\Pi_{y,m}=\Pi_{y,m+1}a_y^*\Pi_{y,m}$ yields the sitewise shift identity
\begin{equation}\label{eq:buffered-projector-shift}
 M_y\Pi_{y,M_N-C_{\rm buf}}
 =\Pi_{y,M_N-C_{\rm buf}+d_y}\,M_y\Pi_{y,M_N-C_{\rm buf}}.
\end{equation}
Choose once and for all $C_{\rm buf}>2d_M+2r_0+2$, and define
\[
 \Pi_x^-:=\prod_{y\in B_{r_0}(x)}
 \mathbf1_{\{n_y\le M_N-C_{\rm buf}\}},\qquad
 \Pi_x^+:=\prod_{y\in B_{r_0}(x)}
 \mathbf1_{\{n_y\le M_N-C_{\rm buf}+d_y\}}.
\]
Multiplying the sitewise identities gives
\begin{equation}\label{eq:multi-site-buffered-projector-shift}
 M\Pi_x^-=\Pi_x^+M\Pi_x^-.
\end{equation}
Every degree-$d_M$ word maps $\operatorname{Ran}\Pi_x^-$ into $\operatorname{Ran}\Pi_x^+$ without commuting a projector through the polynomial.  For a product $XYZ$ of at most three centered factors, insert
$I=\Pi_x^-+(I-\Pi_x^-)$ on the right.  In the term ending in $\Pi_x^-$,
insert $I=\Pi_x^++(I-\Pi_x^+)$ on the left:
\[
 \rho_t(XYZ)=\rho_t(\Pi_x^+XYZ\Pi_x^-)
 +\rho_t((I-\Pi_x^+)XYZ\Pi_x^-)
 +\rho_t(XYZ(I-\Pi_x^-)).
\]
No projector is commuted through the polynomial.

For $X,Y\in\{Q,P\}$, write $\delta X_y=X_{y,N}-\rho_t(X_{y,N})$ and $\delta Y_z=Y_{z,N}-\rho_t(Y_{z,N})$.  Two centered factors based at one site are bounded by $\mathcal C_y(t)$ and at different sites by noncommutative Cauchy--Schwarz:
\begin{equation}\label{eq:cross-site-covariance-bound}
 |\rho_t(\delta X_y\delta Y_z)|
 \le\rho_t(\delta X_y\delta X_y^*)^{1/2}
      \rho_t(\delta Y_z^*\delta Y_z)^{1/2}
 \le C\bigl(\mathcal C_y(t)+\mathcal C_z(t)+N^{-1}\bigr).
\end{equation}
For three centered factors, on $\operatorname{Ran}\Pi_x^+$ one rescaled field has norm at most $C\sqrt c$, while the remaining two are controlled by \eqref{eq:cross-site-covariance-bound}.  On the complement, a fixed-degree rescaled polynomial has norm at most $C_{\mathscr P,c}$ and the union bound plus \eqref{eq:strong-cutoff-boundary-tail} gives an exponentially small probability.  Cauchy--Schwarz turns every bulk--boundary cross term into $C_{\mathscr P,T,c}e^{-c_{\rm bd}N}$ with, for example, $c_{\rm bd}=c_{\rm tail}/2$.  Summing the finitely many normal-ordered monomials proves \eqref{eq:cubic-moment-closure}.
\end{proof}

\begin{lemma}\label[lemma]{lem:fluctuation-propagation}
Fix $T<\infty$ and a volume-uniform compact amplitude class $\mathcal K$.  Choose $c\ge c_T(\mathcal K)$ as in \Cref{lem:cutoff-tail-propagation}.  For projected coherent product initial data with amplitudes in $\mathcal K$, the evolved finite-$N$ states satisfy
\begin{equation}\label{eq:finite-time-fluctuation-bound}
 \sup_{0\le t\le T}\sup_{x\in\Lambda_L}\mathcal C_x(t)
 \le C_{T,c}N^{-1},
\end{equation}
where $C_{T,c}$ is independent of $L$ and $N$.
\end{lemma}
\begin{proof}
Fix a site $x$ and write
\[
 m_x(t)=\rho_t(a_x),\qquad b_x=a_x-m_x(t),\qquad
 v_x(t)=\rho_t(b_x^*b_x).
\]
Away from the exponentially small cutoff boundary,
$\mathcal C_x=(2v_x+1)/N+O(e^{-c_{\rm tail}N})$.  For an individual Heisenberg-generator contribution $\mathcal G$,
\begin{equation}\label{eq:centered-variance-identity}
 \left.\frac d{dt}v_x\right|_{\mathcal G}
 =\rho_t\!\left(\mathcal G(a_x^*a_x)\right)
  -2\operatorname{Re}\!\left[\overline m_x\,\rho_t\!\left(\mathcal G(a_x)\right)\right].
\end{equation}

For one-photon loss, take $\mathcal G=\kappa\cD[a_x]^*$.  Since $\cD[a_x]^*(a_x)=-a_x/2$ and
$\cD[a_x]^*(a_x^*a_x)=-a_x^*a_x$, we obtain
\begin{equation}\label{eq:one-photon-variance}
 \left.\dot v_x\right|_{\kappa\cD[a_x]^*}=-\kappa v_x,
 \qquad
 \left.\dot{\mathcal C}_x\right|_{\kappa\cD[a_x]^*}
 =-\kappa\mathcal C_x+\frac\kappa N+O(e^{-c_{\rm tail}N}).
\end{equation}
A coherent linear drive changes $m_x$, but its contribution to
\eqref{eq:centered-variance-identity} cancels.

Every quadratic Hamiltonian gives a linear centered equation
\[
 \left.\dot b_x\right|_H=\sum_{y\in B_{r_0}(x)}
 (U_{xy}b_y+W_{xy}b_y^*),
\]
where $U_{xy},W_{xy}\in\mathbb C$ vanish for $y\notin B_{r_0}(x)$, are uniformly bounded, and collect the onsite rotation, onsite squeezing, hopping, and bond-pairing coefficients.  Thus
\begin{equation}\label{eq:linear-hamiltonian-variance}
 \left.\dot v_x\right|_H
 =2\operatorname{Re}\sum_y
 \left[U_{xy}\rho_t(b_x^*b_y)+W_{xy}\rho_t(b_x^*b_y^*)\right].
\end{equation}
Noncommutative Cauchy--Schwarz gives
\[
 |\rho_t(b_x^*b_y)|\le\sqrt{v_xv_y},\qquad
 |\rho_t(b_x^*b_y^*)|\le\sqrt{v_x(v_y+1)},
\]
so the right side of \eqref{eq:linear-hamiltonian-variance} is bounded by
$C\max_{z\in B_{r_0}(x)}v_z+C$.

For the unit-rate bond jump $L_{xy}=a_x-a_y$,
\[
 \cD[L_{xy}]^*a_x=-\frac12(a_x-a_y),\qquad
 \cD[L_{xy}]^*(a_x^*a_x)
 =-a_x^*a_x+\frac12(a_y^*a_x+a_x^*a_y).
\]
Equation \eqref{eq:centered-variance-identity} yields
\begin{equation}\label{eq:bond-dissipator-variance}
 \left.\dot v_x\right|_{\cD[a_x-a_y]^*}
 =-v_x+\operatorname{Re}\rho_t(b_x^*b_y)
 \le -\frac12v_x+\frac12v_y.
\end{equation}
Summing the finitely many bonds produces the neighborhood maximum; the fixed commutator in $\mathcal C_x=(2v_x+1)/N$ contributes only $C/N$.

For the nonlinear channel,
\[
 \cD[a_x^2]^*a_x=-a_x^*a_x^2,
 \qquad
 \cD[a_x^2]^*(a_x^*a_x)=-2a_x^{*2}a_x^2.
\]
Substitution of $a_x=m_x+b_x$ into \eqref{eq:centered-variance-identity} gives
\begin{align}
 \left.\dot v_x\right|_{(\gamma/N)\cD[a_x^2]^*}
 =-\frac\gamma N\Bigl[&4|m_x|^2v_x
 +\overline m_x^{\,2}\rho_t(b_x^2)+m_x^2\rho_t(b_x^{*2})\notag\\
 &+3\overline m_x\rho_t(b_x^*b_x^2)
  +3m_x\rho_t(b_x^{*2}b_x)
  +2\rho_t(b_x^{*2}b_x^2)\Bigr].
 \label{eq:two-photon-centered-variance}
\end{align}
The anomalous moments satisfy
$|\rho_t(b_x^2)|\le\sqrt{v_x(v_x+1)}$.  The scale of $m_x$ follows directly from the exponential-moment bound: for a constant $m_T$ independent of $N$, Jensen's inequality gives
\[
 e^{\theta\rho_t(n_x)}\le \rho_t(e^{\theta n_x})\le e^{Nm_T},
 \qquad
 \rho_t(n_x)\le\frac{Nm_T}{\theta},
 \qquad
 |m_x|^2\le\rho_t(n_x).
\]
Hence $|m_x|\le C_T\sqrt N$.  Young's inequality controls the cubic terms in
\eqref{eq:two-photon-centered-variance} by the final nonpositive quartic term and lower moments.  With
\[
 \mathcal Q_x(t)=\frac\gamma{N^2}\rho_t(b_x^{*2}b_x^2)\ge0,
\]
we obtain
\begin{equation}\label{eq:two-photon-covariance-drift}
 \left.\frac d{dt}\mathcal C_x(t)\right|_{\rm 2ph}
 \le C_c\mathcal C_x(t)+\frac{C_c}{N}-\mathcal Q_x(t).
\end{equation}
Cubic centered moments are bounded by \Cref{lem:cubic-moment-closure}.  Normal-ordering defects are $O(N^{-1})$.  Each cutoff discrepancy contains an original-space boundary projector; the cutoff-tail estimate, together with a fixed-degree cutoff norm bound, makes these terms exponentially small in $N$.  Combining the preceding estimates gives
\[
 \frac d{dt}\mathcal C_x(t)
 \le C_{T,c}\max_{z\in B_{r_0}(x)}\mathcal C_z(t)
 +\frac{C_{T,c}}N+C_{T,c}e^{-c_{\rm tail}N}-\mathcal Q_x(t).
\]
Gronwall's lemma, applied to the maximum over $x$ after dropping the favorable term $\mathcal Q_x$, and the $O(N^{-1})$ projected coherent initial covariance prove \eqref{eq:finite-time-fluctuation-bound}.
\end{proof}

\begin{proposition}\label[proposition]{prop:coherent-propagation}
Fix $T<\infty$ and a volume-uniform compact amplitude class $\mathcal K$.  Choose $c\ge c_T(\mathcal K)$ as in \Cref{lem:cutoff-tail-propagation}, and let $E<c$ exceed the bound supplied by \Cref{lem:classical-bound}.  For projected coherent product initial data with amplitudes in $\mathcal K$, the finite-$N$ Heisenberg first moments satisfy, uniformly in $x\in\Lambda_L$ and with constants independent of $L$,
\[
  \sup_{0\le t\le T}\sup_{x\in\Lambda_L}\left(
  \left|\rho_{L,N}(t)(Q_{x,N})-q_x(t)\right|
  +\left|\rho_{L,N}(t)(P_{x,N})-p_x(t)\right|
  \right)\le C_T N^{-1}.
\]
Here $(q_x(t),p_x(t))$ solves the deterministic reaction--transport equation with matching initial coherent amplitudes.
\end{proposition}
\begin{proof}
Write
\[
  e_x(t)=\bigl(\rho_{L,N}(t)(Q_{x,N})-q_x(t),\,
               \rho_{L,N}(t)(P_{x,N})-p_x(t)\bigr).
\]
Let $\mathcal F_\lambda^{\Lambda_L}$ denote the full lattice vector field, including reaction and transport, and let $\mathcal F_{\lambda,q,x}^{\Lambda_L}$ and $\mathcal F_{\lambda,p,x}^{\Lambda_L}$ be its two components at $x$.  A direct Heisenberg commutator calculation yields, for the evolved state,
\[
  \frac{d}{dt}\rho_t(Q_{x,N})
  =\mathcal F_{\lambda,q,x}^{\Lambda_L}\bigl(\rho_t(Q),\rho_t(P)\bigr)
   +\operatorname{Err}^{\rm fluc}_{q,x}(t)
   +\operatorname{Err}^{\rm cut}_{q,x}(t)
   +\operatorname{Err}^{\rm ord}_{q,x}(t),
\]
and similarly for \(P_{x,N}\).  Here $\operatorname{Err}^{\rm fluc}$ collects centered-moment terms, $\operatorname{Err}^{\rm cut}$ collects terms containing a cutoff-boundary projection, and $\operatorname{Err}^{\rm ord}$ collects normal-ordering commutators.  The components of $\mathcal F_\lambda^{\Lambda_L}$ are the coherent symbols identified by \Cref{lem:symbol}; this identifies the deterministic drift, while the error estimate itself is taken on the evolved state.  The local moment expansion in \Cref{lem:cubic-moment-closure} yields
\[
  |\operatorname{Err}^{\rm fluc}_{q,x}(t)|
  +|\operatorname{Err}^{\rm fluc}_{p,x}(t)|
  \le C_T\max_{y\in B_{r_0}(x)}\mathcal C_y(t),
\]
the finite-time cutoff-tail lemma supplies exponentially small cutoff-boundary errors for the chosen $c\ge c_T(\mathcal K)$, and normal-ordering commutators contribute \(O(N^{-1})\).  The cutoff error includes both the boundary commutator in \([a_{M_N},a_{M_N}^*]\) and the difference between projected and unprojected coherent symbols: the former is controlled for evolved states by \Cref{lem:cutoff-tail-propagation}, while the latter is exponentially small at \(t=0\) by \Cref{lem:comm}.  Hence \Cref{lem:fluctuation-propagation} implies
\[
  \frac{d}{dt}\sup_x |e_x(t)|
  \le C_T\sup_x |e_x(t)|+C_TN^{-1}.
\]
The initial first-moment error is \(O(N^{-1})\) by the coherent cutoff-tail estimate of \Cref{lem:comm}.  Gronwall's lemma proves the displayed uniform first-moment estimate.  The constants depend on \(T\), the fixed local polynomial degree, and the compact classical sector, but not on \(L\) or \(N\).
\end{proof}

\subsection{Turing assumptions and Bragg order}
Fix a field dimension $m_{\rm op}\in\mathbb N$ and write the local quantum order parameter as
\[
  O_{x,N}=(O_{x,N}^{1},\ldots,O_{x,N}^{m_{\rm op}}).
\]
Its rescaled first moments are collected in $u_x\in\mathbb R^{m_{\rm op}}$.  In the explicit realization below $m_{\rm op}=2$, $O_{x,N}=(Q_{x,N},P_{x,N})$, and $u_x=(q_x,p_x)$.  Let $\mathscr R_{\rm hop}\subset\mathbb Z^d$ be a fixed finite displacement set, let $F_\lambda:\mathbb R^{m_{\rm op}}\to\mathbb R^{m_{\rm op}}$ be the local reaction field, and let $D_{\xi,\lambda}(u)\in\mathbb R^{m_{\rm op}\times m_{\rm op}}$ for $\xi\in\mathscr R_{\rm hop}$.  We suppose that the first moments close to
\begin{equation}\label{eq:closure}
\dot u_x=\mathcal F_\lambda^{\Lambda_L}(u)_x
:=F_\lambda(u_x)+\sum_{\xi\in\mathscr R_{\rm hop}}
D_{\xi,\lambda}(u_x)(u_{x+\xi}-u_x),
\end{equation}
where addition of lattice sites is understood modulo $L$.  Thus $\mathcal F_\lambda^{\Lambda_L}:(\mathbb R^{m_{\rm op}})^{\Lambda_L}\to(\mathbb R^{m_{\rm op}})^{\Lambda_L}$ is the full lattice drift, as distinguished from the local reaction field $F_\lambda$.

We identify the Brillouin torus with $\mathbb T^d:=(-\pi,\pi]^d$ and the momentum grid of $\Lambda_L$ with $\Lambda_L^*:=(2\pi/L)(\mathbb Z/L\mathbb Z)^d\subset\mathbb T^d$.  With the unitary discrete Fourier convention
\[
  \widehat f(k)=|\Lambda_L|^{-1/2}\sum_{x\in\Lambda_L}e^{-ik\cdot x}f_x,
  \qquad k\in\Lambda_L^*,
\]
let $u_{{\rm hom},\lambda}$ denote the homogeneous stationary state and let $\widehat A_\lambda(k)$ denote the symbol of $D_u\mathcal F_\lambda^{\Lambda_L}(u_{{\rm hom},\lambda})$.  The limiting equation lies in the classical reaction--diffusion and semilinear-parabolic setting \cite{Henry1981,Murray2003}.

\begin{definition}\label[definition]{def:turing-point}
A homogeneous stationary state is a Turing point at a nonempty compact critical set $K_*\subset\mathbb T^d\setminus\{0\}$ when:
\begin{enumerate}[label=(T\arabic*)]
\item $\widehat A_0(k)$ has a simple real zero eigenvalue for every $k\in K_*$;
\item the homogeneous mode $\widehat A_0(0)$ is strictly stable;
\item there are an open neighborhood $U\supset K_*$ and a spectral gap $\delta_{\rm gap}>0$ such that
\[
  \sup_{k\notin U}\max\operatorname{Re}\operatorname{Spec}\widehat A_0(k)\le-\delta_{\rm gap},
\]
while inside $U$ the critical eigenvalue vanishes only on $K_*$ and has a nondegenerate quadratic maximum in every direction normal to $K_*$.  For a finite symmetry orbit, this last condition is understood in all momentum directions.
\end{enumerate}
\end{definition}

In the explicit design below, the determinant has a quadratic zero as a function of the scalar lattice symbol $\omega$, although the zero eigenvalue of the matrix is simple.

\subsubsection{Reaction--transport equation}
We consider two-component cubic reaction--transport equations for $u_x=(q_x,p_x)$ of the form
\begin{equation}\label{eq:rtclass}
\begin{aligned}
\dot q_x&=a q_x+\Omega_\lambda p_x-\nu(q_x^2+p_x^2)q_x
 +\sum_{j=1}^d D_{q,j}\Delta_jq_x,\\
\dot p_x&=-\Omega_\lambda q_x-bp_x-\nu(q_x^2+p_x^2)p_x
 +\sum_{j=1}^d D_{p,j}\Delta_jp_x,
\end{aligned}
\end{equation}
with $b>a>0$, $D_{q,j},D_{p,j}>0$, and $\nu>0$.  We assume that $\lambda\mapsto\Omega_\lambda>0$ is real analytic near $0$.  For the stripe, $(D_{q,1},D_{p,1})=(D_q,D_p)$ and $D_{q,j}=D_{p,j}=D_y$ for $j\ge2$.  The isotropic spot and labyrinth examples use $D_{q,j}=D_q$ and $D_{p,j}=D_p$ in every direction.
\begin{lemma}\label[lemma]{lem:classical-bound}
Every solution of \eqref{eq:rtclass} with uniformly bounded initial data is global.  If
\[
  M(t):=\max_{x\in\Lambda_L}\bigl(q_x(t)^2+p_x(t)^2\bigr),
  \qquad
  C_{\rm lin}:=2a+2\sum_{j=1}^d\max\{D_{q,j},D_{p,j}\},
\]
then its upper Dini derivative satisfies
\[
  D^+M(t)\le C_{\rm lin}M(t)-2\nu M(t)^2.
\]
Scalar comparison yields
\[
  M(t)\le\max\left\{M(0),\frac{C_{\rm lin}}{2\nu}\right\}
\]
for all $t\ge0$, uniformly in $L$.
\end{lemma}
\begin{proof}
At a site where $q_x^2+p_x^2=M(t)$, the $\Omega_\lambda$ terms cancel.  Using $2uv\le u^2+v^2$ for the two neighbors gives
\[
  2q_x\Delta_jq_x\le2\bigl(M(t)-q_x^2\bigr)=2p_x^2,
  \qquad
  2p_x\Delta_jp_x\le2\bigl(M(t)-p_x^2\bigr)=2q_x^2.
\]
Hence
\[
  2D_{q,j}q_x\Delta_jq_x+2D_{p,j}p_x\Delta_jp_x
  \le2\max\{D_{q,j},D_{p,j}\}M(t).
\]
Dropping the favorable term $-2bp_x^2$ leaves a logistic differential inequality.  Scalar comparison yields the displayed bound and rules out finite-time blow-up.
\end{proof}

For a commensurate fixed-space reduction, let $\mathcal C$ be the chosen finite period cell, write $\mathcal F_\lambda^{\mathcal C}$ for the periodic restriction of \eqref{eq:closure}, and let $\mathcal X_{\rm fix}\subset(\mathbb R^{m_{\rm op}})^{\mathcal C}$ be a real symmetry or isotropy fixed subspace with its normalized Euclidean inner product.  Set
\[
  \mathcal A_*:=D_u\mathcal F_0^{\mathcal C}(u_{{\rm hom},0})
  \big|_{\mathcal X_{\rm fix}}.
\]
Whenever $\ker\mathcal A_*$ is one-dimensional, choose a direct sum
$\mathcal X_{\rm fix}=\ker\mathcal A_*\oplus\mathcal X_{\rm ran}$ and denote the associated projections by $\Pi_{\rm ker}$ and $\Pi_{\rm ran}$.

\begin{assumption}\label[assumption]{ass:turing}
Choose $k_*\in K_*$.  At the nonzero critical set $K_*$, assume:
\begin{enumerate}[label=(A\arabic*)]
\item the homogeneous state is a Turing point in the sense of \Cref{def:turing-point};
\item the operator $\mathcal A_*$ has one-dimensional kernel and $\mathcal A_*|_{\mathcal X_{\rm ran}}$ is invertible;
\item the parameter $\lambda$ unfolds the simple critical eigenvalue with positive crossing slope in that fixed space.
\end{enumerate}
When the critical set is a continuous shell, the fixed space in (A2) is chosen before the one-dimensional reduction.  For an equation of the form \eqref{eq:rtclass}, \Cref{lem:cubic-coefficient} gives supercritical saturation.
\end{assumption}
Choose right and left critical vectors $r,\ell\in\mathbb R^{m_{\rm op}}$ with the normalization
\[
\widehat A_0(k_*)r=0,\qquad \ell^{\mathsf T}\widehat A_0(k_*)=0,
\qquad \ell^{\mathsf T}r=1.
\]
The crossing slope used below is the normalization-independent quantity
\begin{equation}\label{eq:crossing-slope}
  \chi_*:=\frac{\ell^{\mathsf T}
  (\partial_\lambda\widehat A_\lambda(k_*)|_{\lambda=0})r}
  {\ell^{\mathsf T}r}>0;
\end{equation}
with the preceding normalization, the denominator equals one.

\begin{assumption}\label[assumption]{ass:bragg}
For the Bragg estimate, suppose that the reaction--transport equation has the form \eqref{eq:rtclass}, satisfies the Turing hypotheses in \Cref{ass:turing}, and has critical kernel
\[
  \operatorname{span}\{r\psi_*\}
\]
in the chosen real symmetry-fixed space.  Define the period lattice, its finite translation group, and its dual by
\[
  \Gamma_{\mathcal C}:=\{\gamma\in\mathbb Z^d:\psi_*(x+\gamma)=\psi_*(x)
  \text{ for every }x\},
  \qquad
  G_{\mathcal C}:=\mathbb Z^d/\Gamma_{\mathcal C},
\]
\[
  G_{\mathcal C}^*:=\{k\in\mathbb T^d:e^{ik\cdot\gamma}=1
  \text{ for every }\gamma\in\Gamma_{\mathcal C}\}.
\]
We further assume:
\begin{enumerate}[label=(G\arabic*)]
\item $\Gamma_{\mathcal C}$ has finite index, and a fixed fundamental period cell $\mathcal C$ tiles a sequence $\Lambda_{L_j}$ of compatible tori with $L_j\to\infty$;
\item there is a nonzero momentum $k\in K_*\cap G_{\mathcal C}^*$ for which
\[
  \widehat\psi_*(k)=|\mathcal C|^{-1}\sum_{x\in\mathcal C}e^{-ik\cdot x}\psi_*(x)
\]
is nonzero;
\item there is a local pattern covector $v_{\rm pat}\in(\mathbb R^{m_{\rm op}})^*$ satisfying $v_{\rm pat}\cdot r\ne0$.
\end{enumerate}
\end{assumption}
Let $u_{\lambda,0}$ denote a chosen fixed-space branch representative supplied by \Cref{thm:reflection-bifurcation}.  Let $H_{\mathcal C}\le G_{\mathcal C}$ be its stabilizer and let $\Theta:=G_{\mathcal C}/H_{\mathcal C}$ be its finite translation orbit.  For $\tau\in\Theta$, represented by any element of $G_{\mathcal C}$, set
\begin{equation}\label{eq:translation-orbit-definitions}
  \psi_{*,\tau}(x):=\psi_*(x-\tau),\qquad
  u_{\lambda,\tau}(x):=u_{\lambda,0}(x-\tau).
\end{equation}
These definitions are independent of the chosen coset representative because $H_{\mathcal C}$ stabilizes the branch.

\begin{theorem}[Lindblad realization and Bragg order]\label[theorem]{thm:main}
Suppose the Bragg assumptions in \Cref{ass:bragg} hold.  Then the reaction--transport equation admits a family of Lindblad realizations with finite-range couplings, noncommuting order parameters, and the prescribed first-moment limit.  Let $u_{\lambda,\tau}$ be a translate of the nonlinear branch and initialize the microscopic system in the projected coherent product state associated with $u_{\lambda,\tau}$.  On every bounded time interval $[0,T]$, there is a cutoff threshold $c_T<\infty$ such that, for each $c\ge c_T$, the quantum structure-factor lower bound in \Cref{cor:dynamical-quantum-bragg} persists under the Lindblad evolution.  The ratio $c$ controls the finite-$N$ Fock cutoff and does not enter the limiting reaction--transport equation.
\end{theorem}
\begin{proof}
The bifurcation theorem supplies the nonlinear branch and \Cref{lem:general-bragg} its deterministic structure-factor lower bound.  The local parameter map in \Cref{thm:lindblad} realizes the reaction--transport equation, and \Cref{lem:quantum-bragg-lift} transfers the bound to projected coherent product states.  For a fixed $T$, choose $c\ge c_T$ from \Cref{lem:cutoff-tail-propagation}; the first-moment estimate \Cref{prop:coherent-propagation} then yields \Cref{cor:dynamical-quantum-bragg}.
\end{proof}

For finite $L,N$, the local order parameters are $O_{x,N}=(Q_{x,N},P_{x,N})$, with $N[Q_{x,N},P_{x,N}]\to i$ on bounded coherent-state sectors.  The right critical eigenvector determines the local quadrature combination carrying the emerging pattern, whereas the left critical eigenvector determines the soft Heisenberg observable; in this paper it is
\[
  \mathcal O_x^{\rm soft}=\ell\cdot O_{x,N},\qquad
  \widehat{\mathcal O}^{\rm soft}(k_*)=|\Lambda_L|^{-1/2}\sum_xe^{-ik_*\cdot x}\mathcal O_x^{\rm soft}.
\]

\section{Construction and stability of a commensurate Turing stripe}\label{sec:branch}
\subsection{Spectral design and commensurate branches}

\begin{theorem}\label[theorem]{thm:design}
Let
\[
J_{\rm reac}=\begin{pmatrix}a&\Omega\\-\Omega&-b\end{pmatrix},\qquad b>a>0,
\]
and choose one lattice direction with transport coefficients $D_q,D_p>0$.  Assume
\[
0<\omega_*:=\frac{aD_p-bD_q}{2D_qD_p}<4,
\qquad
\Omega^2=ab+\frac{(aD_p-bD_q)^2}{4D_qD_p}.
\]
Then the linearized system has a finite-wave-number Turing point at $2(1-\cos k_*)=\omega_*$ along that direction, and its determinant has a quadratic zero in the scalar symbol $\omega$.  The finite-$k$ instability disappears when $D_q=D_p$.
\end{theorem}
\begin{proof}
For $\omega=2(1-\cos k)$,
\[
\det(J_{\rm reac}-\omega\diag(D_q,D_p))=D_qD_p\omega^2+(bD_q-aD_p)\omega+(\Omega^2-ab).
\]
The chosen value of $\Omega$ makes the discriminant zero and places the vertex at $\omega_*$.  The trace is negative there, so the zero eigenvalue is simple even though the determinant has second-order contact as a function of $\omega$.  Moreover $\operatorname{tr}J_{\rm reac}=a-b<0$ and $\det J_{\rm reac}=\Omega^2-ab=(aD_p-bD_q)^2/(4D_qD_p)>0$, so the homogeneous reaction matrix is stable.  The condition $0<\omega_*<4$ gives a lattice wave number.  If $D_q=D_p=D$, then $J_{\rm reac}-\omega DI$ shifts every eigenvalue of $J_{\rm reac}$ to the left, and scalar transport cannot destabilize the stable matrix $J_{\rm reac}$ at finite $k$.
\end{proof}

If the pair $(D_q,D_p)$ acts in every direction, the symbol depends on $\omega_{\rm tot}(k)=\sum_j2(1-\cos k_j)$ and the critical set is the isotropic shell $\omega_{\rm tot}(k)=\omega_*$.  If differential transport acts only in direction $1$ while the remaining directions have common scalar transport, the transverse scalar terms shift those modes to the left and isolate the longitudinal stripe orbit.

On the fixed period cell $\mathcal C$, use the normalized Euclidean inner product and average
\[
  \langle f,g\rangle_{\mathcal C}:=\frac1{|\mathcal C|}\sum_{x\in\mathcal C}f(x)\cdot g(x),
  \qquad
  \langle f\rangle_{\mathcal C}:=\frac1{|\mathcal C|}\sum_{x\in\mathcal C}f(x),
\]
where the dot is omitted for scalar fields.

\begin{lemma}\label[lemma]{lem:cubic-coefficient}
Assume that the critical kernel in the chosen real symmetry-fixed space is
\[
  \operatorname{span}\{r\psi_*\},
\]
where $r,\ell\in\mathbb R^2$, $\ell^{\mathsf T}r=1$, and $\psi_*$ is a nonzero real critical spatial eigenfunction.  If the crossing slope is $\chi_*>0$, then the reduced equation has the form
\[
  0=B\left(
    \chi_*\lambda-c_3 B^2
    +O(\lambda^2+\lambda B^2+B^4)
  \right),
\]
where
\[
  c_3
  =\nu |r|^2
  \frac{\langle\psi_*,\psi_*^3\rangle_{\mathcal C}}
       {\langle\psi_*,\psi_*\rangle_{\mathcal C}}
  =\nu |r|^2
  \frac{\langle\psi_*^4\rangle_{\mathcal C}}
       {\langle\psi_*^2\rangle_{\mathcal C}}>0.
\]
For a single complex Fourier pair normalized as
\[
  u_{\rm crit}
  =B r e^{ik_*\cdot x}+\overline B r e^{-ik_*\cdot x},
\]
the resonant complex amplitude equation instead has cubic term
$-3\nu|r|^2B|B|^2$, so this specialization is $c_3=3\nu|r|^2$ in the complex Fourier-coefficient convention.
\end{lemma}
\begin{proof}
The nonlinearity has no quadratic part, so the range correction enters the critical projection only at the displayed higher orders.  For a real fixed-space critical component $B r\psi_*$,
\[
  -\nu|B r\psi_*|^2(B r\psi_*)
  =-\nu B^3|r|^2r\psi_*^3.
\]
Projection onto the left critical vector and spatial eigenfunction, with $\ell^{\mathsf T}r=1$, produces the displayed overlap coefficient.  Positivity follows because $\psi_*$ is real and nonzero.  For the complex pair, the coefficient of $e^{ik_*\cdot x}$ in
$|u_{\rm crit}|^2u_{\rm crit}$ is $3|B|^2Br|r|^2$, which yields the Fourier-pair specialization.
\end{proof}

\begin{theorem}\label[theorem]{thm:reflection-bifurcation}
Under the Turing assumptions in \Cref{ass:turing}, an equation of the form \eqref{eq:rtclass} has real-analytic maps
\[
  s\longmapsto \bigl(u(s),\lambda(s)\bigr)
\]
near $s=0$, with
\[
  u(s)=s\,r\psi_*+O(s^3),
  \qquad
  \lambda(s)=\frac{c_3}{\chi_*}s^2+O(s^4),
\]
where
\[
  c_3=\nu |r|^2
  \frac{\langle\psi_*^4\rangle_{\mathcal C}}{\langle\psi_*^2\rangle_{\mathcal C}}>0.
\]
The remainder is uniform in the period-cell $\ell^\infty$ norm.  For $\lambda>0$, inversion on either signed branch gives
\[
  u_\lambda^{\pm}=\pm\sqrt{\frac{\chi_*}{c_3}}\,\lambda^{1/2}r\psi_*+O(\lambda^{3/2}),
  \qquad
  B_\lambda^{\pm}=\pm\sqrt{\frac{\chi_*}{c_3}}\,\lambda^{1/2}+O(\lambda^{3/2}).
\]
Its leading Fourier support lies in $K_*$.  If the equation satisfies \Cref{ass:bragg}, every local pattern field with nonzero critical projection has the extensive semiclassical Bragg peak of \Cref{lem:general-bragg}.
\end{theorem}
\begin{proof}
Write the solution as $sr\psi_*+z(s,\lambda)$ in the critical direction plus its range complement.  Since $\mathcal F_\lambda^{\mathcal C}(0)=0$ and the nonlinearity has no quadratic term, the analytic range equation yields
\[
  z(s,\lambda)=O(|\lambda s|+|s|^3).
\]
After division by the critical amplitude, \Cref{lem:cubic-coefficient} becomes
\[
  0=\chi_*\lambda-c_3 s^2+O(\lambda^2+\lambda s^2+s^4).
\]
The analytic implicit-function theorem solves this equation for $\lambda=\lambda(s)$ because $\chi_*>0$.  Substitution yields the branch, and the Bragg estimate follows from \Cref{ass:bragg,lem:general-bragg}.
\end{proof}

For the pattern covector $v_{\rm pat}$, define the deterministic branch field by
\[
  \mathcal O_{\lambda,\tau}^{\rm pat}(x):=v_{\rm pat}\cdot u_{\lambda,\tau}(x).
\]
The $O(\lambda^{3/2})$ remainder in the branch expansion below is uniform in the period-cell $\ell^\infty$ norm, and hence on every period-compatible torus obtained by repetition of $\mathcal C$.
For every period-compatible torus $\Lambda_L$, define
\[
  \widehat{\mathcal O}_{\lambda,\tau}^{\rm pat}(k)
  =|\Lambda_L|^{-1/2}\sum_{x\in\Lambda_L}e^{-ik\cdot x}
  \mathcal O_{\lambda,\tau}^{\rm pat}(x),
\]
\[
  S_{\lambda}^{\rm coh}(k)
  =\frac1{|\Theta|}\sum_{\tau\in\Theta}
  \left|\widehat{\mathcal O}_{\lambda,\tau}^{\rm pat}(k)\right|^2 .
\]
The quantum operator structure factor is introduced in \Cref{lem:quantum-bragg-lift}.
Because the momentum in (G2) belongs to $G_{\mathcal C}^*$, $\tau\mapsto e^{-ik\cdot\tau}$ is a character of $G_{\mathcal C}$.  If $h\in H_{\mathcal C}$, invariance of $\psi_*$ under $h$ gives
\[
  \widehat\psi_*(k)=e^{-ik\cdot h}\widehat\psi_*(k).
\]
The coefficient is nonzero by (G2), so $e^{-ik\cdot h}=1$.  Therefore
\[
  \vartheta_k(\tau):=e^{-ik\cdot\tau}
\]
is a well-defined nontrivial character of $\Theta$, and
\[
  \widehat\psi_{*,\tau}(k)=\vartheta_k(\tau)\widehat\psi_*(k),
  \qquad
  \frac1{|\Theta|}\sum_{\tau\in\Theta}\vartheta_k(\tau)=0.
\]

\begin{lemma}\label[lemma]{lem:general-bragg}
Assume the Bragg conditions and let $\operatorname{Orb}(u_\lambda)$ denote the translation orbit
\[
  u_{\lambda,\tau}(x)
  =B_\lambda r\psi_{*,\tau}(x)+O(\lambda^{3/2}),
  \qquad \tau\in\Theta.
\]
For the momentum $k$ in (G2), set
\[
  c_{\rm Br}(k):=\frac{\chi_*}{c_3}|\widehat\psi_*(k)|^2>0.
\]
Then, for the covector $v_{\rm pat}$ in (G3) and every period-compatible torus $\Lambda_{L_j}$,
\[
  \frac{S_\lambda^{\rm coh}(k)}{|\Lambda_{L_j}|}
  \ge c_{\rm Br}(k)|v_{\rm pat}\cdot r|^2\lambda-C\lambda^2.
\]
For all sufficiently small $\lambda>0$, the right side is at least
$\frac12c_{\rm Br}(k)|v_{\rm pat}\cdot r|^2\lambda$.  No off-diagonal translation-average assumption is required.
\end{lemma}
\begin{proof}
Fourier orthogonality at the fixed dual-lattice momentum $k$ yields, uniformly on the translation orbit,
\[
  \widehat{\mathcal O}_{\lambda,\tau}^{\rm pat}(k)
  =|\Lambda_{L_j}|^{1/2}B_\lambda(v_{\rm pat}\cdot r)
    \vartheta_k(\tau)\widehat\psi_*(k)
    +O(|\Lambda_{L_j}|^{1/2}\lambda^{3/2}).
\]
Translation-orbit averaging changes no diagonal weight, because $|\vartheta_k|=1$, and
\[
  \frac{S_\lambda^{\rm coh}(k)}{|\Lambda_{L_j}|}
  =B_\lambda^2|v_{\rm pat}\cdot r|^2|\widehat\psi_*(k)|^2+O(\lambda^2).
\]
Using $B_\lambda^2=(\chi_*/c_3)\lambda+O(\lambda^2)$ identifies the coefficient.  The absence of off-diagonal terms comes from evaluating one fixed Fourier momentum on a compatible torus.
\end{proof}

\begin{lemma}\label[lemma]{lem:quantum-bragg-lift}
Let $u_{\lambda,\tau}$ be the translation family in \Cref{lem:general-bragg}, write
$\alpha_{\lambda,\tau,x}=(q_{\lambda,\tau,x}+ip_{\lambda,\tau,x})/\sqrt2$, and let
\[
  \rho^{\rm coh}_{L,N,\lambda,\tau}
  =|\psi_{\boldsymbol\alpha_{\lambda,\tau}}^{(L,N)}\rangle
   \langle\psi_{\boldsymbol\alpha_{\lambda,\tau}}^{(L,N)}|
\]
be the associated projected coherent product state.  Define the local quantum pattern observable and its Fourier mode by
\[
  \mathcal O_{x,N}^{\rm pat}:=v_{\rm pat}\cdot O_{x,N},\qquad
  \widehat{\mathcal O}_{L,N}^{\rm pat}(k)
  :=|\Lambda_L|^{-1/2}\sum_{x\in\Lambda_L} e^{-ik\cdot x}\mathcal O_{x,N}^{\rm pat}.
\]
\[
  \mathsf S^{\rm q}_{L,N,\lambda}(k)
  =\frac1{|\Theta|}\sum_{\tau\in\Theta}
  \rho^{\rm coh}_{L,N,\lambda,\tau}\!\left(
    \widehat{\mathcal O}_{L,N}^{\rm pat}(k)^*
    \widehat{\mathcal O}_{L,N}^{\rm pat}(k)
  \right).
\]
There are constants $C,c_0>0$, independent of $j$, $L_j$, and $N$, such that, for every compatible $L_j$,
\[
  \frac{\mathsf S^{\rm q}_{L_j,N,\lambda}(k)}{|\Lambda_{L_j}|}
  \ge c_{\rm Br}(k)|v_{\rm pat}\cdot r|^2\lambda-C\lambda^2
      -C\bigl(N^{-1}+Ne^{-c_0N}\bigr).
\]
Consequently,
\[
  \liminf_{j\to\infty}\liminf_{N\to\infty}
  \frac{\mathsf S^{\rm q}_{L_j,N,\lambda}(k)}{|\Lambda_{L_j}|}
  \ge c_{\rm Br}(k)|v_{\rm pat}\cdot r|^2\lambda-C\lambda^2.
\]
\end{lemma}
\begin{proof}
In this proof write
\[
 \widehat{\mathcal O}:=\widehat{\mathcal O}_{L,N}^{\rm pat}(k),
 \qquad
 e_x:=\rho^{\rm coh}_{L,N,\lambda,\tau}
 \!\left(\mathcal O_{x,N}^{\rm pat}\right)
 -\mathcal O_{\lambda,\tau}^{\rm pat}(x).
\]
Variance positivity yields, for every $\tau\in\Theta$,
\[
  \rho^{\rm coh}_{L,N,\lambda,\tau}
  \bigl(\widehat{\mathcal O}^*\widehat{\mathcal O}\bigr)
  \ge
  \left|\rho^{\rm coh}_{L,N,\lambda,\tau}
  (\widehat{\mathcal O})\right|^2.
\]
By \Cref{lem:symbol,lem:comm}, the projected coherent expectation of each local pattern field equals the deterministic branch profile with a uniform error
$O(N^{-1}+Ne^{-c_0N})$.  The defined sitewise errors satisfy
\[
  \left|\frac1{|\Lambda_L|}\sum_x e_xe^{-ik\cdot x}\right|\le\sup_x|e_x|,
\]
so Fourier summation introduces no volume-dependent loss.  Division by $|\Lambda_L|$, translation-orbit averaging, and \Cref{lem:general-bragg} prove the first inequality; the second follows by taking $N\to\infty$.
\end{proof}

\begin{corollary}\label[corollary]{cor:dynamical-quantum-bragg}
Under the hypotheses of \Cref{lem:quantum-bragg-lift}, let
\[
  \rho_{L,N,\lambda,\tau}(t)
  =e^{t\cL_{L,N}^*}\rho^{\rm coh}_{L,N,\lambda,\tau}
\]
be the actual Lindblad evolution and define
\[
  \mathsf S^{\rm q}_{L,N,\lambda}(k,t)
  =\frac1{|\Theta|}\sum_{\tau\in\Theta}
  \rho_{L,N,\lambda,\tau}(t)\!\left(
    \widehat{\mathcal O}_{L,N}^{\rm pat}(k)^*
    \widehat{\mathcal O}_{L,N}^{\rm pat}(k)
  \right).
\]
For every fixed $T<\infty$ there exists $c_T<\infty$ such that, for every cutoff ratio $c\ge c_T$, there are constants $C_T,c_0>0$, independent of $L$ and $N$, such that, for every compatible $L_j$,
\[
  \inf_{0\le t\le T}
  \frac{\mathsf S^{\rm q}_{L_j,N,\lambda}(k,t)}{|\Lambda_{L_j}|}
  \ge c_{\rm Br}(k)|v_{\rm pat}\cdot r|^2\lambda-C\lambda^2
      -C_T\bigl(N^{-1}+Ne^{-c_0N}\bigr).
\]
The extensive lower bound in \Cref{lem:quantum-bragg-lift} therefore holds after taking $N\to\infty$, uniformly for $0\le t\le T$, and hence also along every compatible thermodynamic sequence.
\end{corollary}
\begin{proof}
For the compact finite branch orbit and prescribed $T$, choose $c\ge c_T$ from \Cref{lem:cutoff-tail-propagation}.  Since the deterministic branch is stationary, \Cref{prop:coherent-propagation} and \Cref{lem:comm,lem:symbol} give
\[
  \sup_{0\le t\le T,x}
  \left|\rho_{L,N,\lambda,\tau}(t)(\mathcal O_{x,N}^{\rm pat})
        -\mathcal O_{\lambda,\tau}^{\rm pat}(x)\right|
  \le C_T\bigl(N^{-1}+Ne^{-c_0N}\bigr).
\]
Variance positivity at time $t$, followed by Fourier summation and \Cref{lem:general-bragg}, gives the lower bound.
\end{proof}

\begin{corollary}\label[corollary]{cor:orbit-averaged-bragg}
For the translation-orbit-averaged state
\[
  \overline\rho_{L,N,\lambda}(t)
  =\frac1{|\Theta|}\sum_{\tau\in\Theta}
    \rho_{L,N,\lambda,\tau}(t),
\]
define
\[
  \mathsf S^{\rm q}_{\rm conn}(k,t)
  =\overline\rho_{L,N,\lambda}(t)
     \!\left(\widehat{\mathcal O}_{L,N}^{\rm pat}(k)^*
     \widehat{\mathcal O}_{L,N}^{\rm pat}(k)\right)
   -\left|\overline\rho_{L,N,\lambda}(t)
     \!\left(\widehat{\mathcal O}_{L,N}^{\rm pat}(k)\right)\right|^2.
\]
For the translation-covariant unseeded generator, the nontrivial character $\vartheta_k$ has zero orbit average.  Hence
\[
 \overline\rho_{L,N,\lambda}(t)
 \!\left(\widehat{\mathcal O}_{L,N}^{\rm pat}(k)\right)=0,
 \qquad
 \mathsf S^{\rm q}_{\rm conn}(k,t)
 =\mathsf S^{\rm q}_{L,N,\lambda}(k,t).
\]
The connected Bragg correlation is therefore extensive.  This operator statement is not by itself an entanglement criterion; entanglement is analyzed separately through the fluctuation state.
\end{corollary}
\begin{proof}
Translation covariance carries translated branch states into translated evolved states.  Averaging the nonzero-momentum one-point function over the nontrivial character $\vartheta_k$ gives zero, while linearity preserves the orbit-averaged two-point function.
\end{proof}

For the explicit construction,
\begin{align}
\dot q_x&=q_x+\Omega_\lambda p_x-\nu(q_x^2+p_x^2)q_x+D_{q,1}\Delta_1q_x+D_y\Delta_2q_x,\label{eq:q}\\
\dot p_x&=-\Omega_\lambda q_x-3p_x-\nu(q_x^2+p_x^2)p_x+D_{p,1}\Delta_1p_x+D_y\Delta_2p_x,\label{eq:p}
\end{align}
where \(D_y>0\) is fixed arbitrarily; writing $D_q:=D_{q,1}$ and $D_p:=D_{p,1}$,
\[
D_q=1,\qquad D_p=3+2\sqrt3,\qquad \Omega_\lambda=\sqrt{2\sqrt3-\lambda}>0.
\]
Then $k_*=(\pi/6,0)$.  Along the longitudinal period-$12$ cell, let
\[
  C_{12}:=\mathbb Z/12\mathbb Z,\qquad j\in C_{12},
\]
and define the unit translation and the site- and bond-centered reflections by
\[
  T_{\rm lat}j=j+1,\qquad \mathscr R_sj=-j,\qquad \mathscr R_bj=-1-j
  \quad (\bmod\ 12).
\]
They act on fields by pullback, $(g\cdot u)(j)=u(g^{-1}j)$.

\begin{theorem}[Commensurate stripe branches]\label[theorem]{thm:branch}
At the Turing point $k_*=(\pi/6,0)$, the homogeneous state of \eqref{eq:q}--\eqref{eq:p} undergoes a supercritical stationary bifurcation.  The site- and bond-reflection-fixed spaces each have a one-dimensional critical kernel and an invertible range operator.  They contain real-analytic branches parameterized by the signed critical amplitude $s$,
\[
  u^{\rm site}(s,j)=2s\,r\cos(k_*j)+O(s^3),\qquad
  u^{\rm bond}(s,j)=2s\,r\cos(k_*j+\pi/12)+O(s^3),
\]
\[
  \lambda(s)=8\sqrt3\,\nu s^2+O(s^4).
\]
For $\lambda>0$, the site- and bond-centered representatives and their lattice translates have leading harmonics
\[
\binom{q_{\lambda,\mu}(x)}{p_{\lambda,\mu}(x)}
 =2B_\lambda r\cos(k_*\cdot x+\phi_0+\mu\pi/6)+O(\lambda^{3/2}),
\qquad \mu\in C_{12},
\]
where
\[
  r=\binom{r_q}{r_p}=\binom{1}{-\beta},\qquad
  \beta=\frac{\sqrt3-1}{\sqrt{2\sqrt3}},\qquad
  B_\lambda=\sqrt{\frac{\lambda}{8\sqrt3\,\nu}}+O(\lambda^{3/2}).
\]
Here $\phi_0=0$ for the site-centered representative and $\phi_0=\pi/12$ for the bond-centered representative.
\end{theorem}
The translate indexed by $\mu\in C_{12}$ is fixed by $T_{\rm lat}^{\mu}\mathscr R_sT_{\rm lat}^{-\mu}$ or $T_{\rm lat}^{\mu}\mathscr R_bT_{\rm lat}^{-\mu}$, according to the representative.
\begin{proof}
For $k_2=0$,
\[
\det \widehat A_0(k_1,0)
 =(1-\omega)(-3-(3+2\sqrt3)\omega)+2\sqrt3
 =(3+2\sqrt3)(\omega-(2-\sqrt3))^2.
\]
At the homogeneous mode,
\[
\widehat A_0(0)=\begin{pmatrix}1&\sqrt{2\sqrt3}\\-\sqrt{2\sqrt3}&-3\end{pmatrix},
\qquad
\operatorname{tr}\widehat A_0(0)=-2,
\qquad
\det \widehat A_0(0)=2\sqrt3-3>0.
\]
The zero wave number is strictly stable.  Transverse scalar diffusion shifts every mode with $k_2\ne0$ strictly to the left, so the anisotropic critical orbit is $\{k_*,-k_*\}$.  At $k_*$ the trace is $-4$, the crossing slope is $\chi_*=1/4$, and right and left critical vectors may be chosen as
\[
  r=(1,-\beta)^{\mathsf T},\qquad
  \ell=\frac{(1,\beta)^{\mathsf T}}{1-\beta^2},
  \qquad \ell^{\mathsf T}r=1.
\]

The period cell is equivariant under the dihedral group $D_{12}$ generated by $T_{\rm lat}$ and the two reflections defined above.  In $\operatorname{Fix}(\mathscr R_s)$ the critical eigenspace is generated by $r\cos(k_*j)$; in $\operatorname{Fix}(\mathscr R_b)$ it is generated by $r\cos(k_*(j+1/2))$.  The corresponding sine mode is odd and is absent.  All remaining period-$12$ Fourier modes, including $3k_*=(\pi/2,0)$, are separated from zero, and the transverse modes remain strictly stable.  The standard simple-eigenvalue and equivariant Lyapunov--Schmidt argument applies \cite{CrandallRabinowitz1971,IoossJoseph1980,GolubitskyStewartSchaeffer1988,Sattinger1979}; the range equation is solved analytically by the implicit-function theorem.

Writing
\[
  u(x)=B r e^{ik_*\cdot x}+\overline B r e^{-ik_*\cdot x}+w,
  \qquad
  w\perp \operatorname{span}\{r e^{ik_*\cdot x},r e^{-ik_*\cdot x}\},
\]
where $\perp$ refers to the normalized period-cell $\ell^2$ inner product, extended sesquilinearly to complex Fourier modes.  The range correction satisfies $w=O(|B|^3+\lambda|B|)$.  Projection by the normalized left vector $\ell^{\mathsf T}$ gives the resonant cubic term
\[
  -3\nu(1+\beta^2)|B|^2B=-2\sqrt3\,\nu |B|^2B.
\]
For $p,q\in\mathbb N_0$, a monomial $B^p\overline B^{q}$ is $C_{12}$-equivariant precisely when $p-q\equiv1\pmod{12}$.  Denote by $\sigma_{\rm lock}\in\mathbb R$ the coefficient of the first phase-locking monomial.  The reduced stationary equation is
\[
0=B\left(\frac{\lambda}{4}-2\sqrt3\,\nu |B|^2
       +O(\lambda^2+\lambda |B|^2+|B|^4)\right)
  +\sigma_{\rm lock}\overline B^{11}+O(|B|^{13}+\lambda |B|^{11}).
\]
The selection rule makes $\overline B^{11}$ the first symmetry-allowed locking term.  On either real reflection-fixed subspace this yields
\[
  \lambda(s)=8\sqrt3\,\nu s^2+O(s^4),
  \qquad
  B_\lambda=\sqrt{\frac{\lambda}{8\sqrt3\,\nu}}+O(\lambda^{3/2}).
\]
Equivariance produces the translated representatives and their conjugate reflections.
\end{proof}

\begin{theorem}[Stability in the reflection-fixed spaces]\label[theorem]{thm:reflection-stability}
Let $\mathcal X_{\rm site}$ and $\mathcal X_{\rm bond}$ be the site- and bond-reflection fixed subspaces of the real period-$12$ first-moment phase space.  For every sufficiently small $\lambda>0$, the corresponding branch $u_\lambda^{\rm site}$ or $u_\lambda^{\rm bond}$ is locally asymptotically stable for the period-cell first-moment dynamics restricted to its fixed space.  More precisely, there are $\lambda_0,c_*,g_*>0$ such that for $0<\lambda<\lambda_0$ the linearization at either branch has one center-originating eigenvalue
\begin{equation}\label{eq:reflection-radial-eigenvalue}
 \mu_{\rm rad}(\lambda)=-\frac\lambda2+O(\lambda^2),
\end{equation}
and every remaining fixed-space eigenvalue satisfies $\operatorname{Re}\mu\le-g_*$.  In particular the fixed-space spectral abscissa is at most $-c_*\lambda$.
\end{theorem}
\begin{proof}
On either reflection-fixed space, the critical eigenspace at $\lambda=0$ is one-dimensional and all range eigenvalues lie in $\{\operatorname{Re}z\le-g_0\}$ for some $g_0>0$; this is the spectral decomposition used in \Cref{thm:branch}.  The center-manifold coordinate may be chosen as the signed real critical amplitude $B$.  The projected coefficients of the dynamical reduced equation coincide with those of the stationary bifurcation equation,
\begin{equation}\label{eq:reflection-amplitude-dynamics}
 \dot B=\frac\lambda4B-2\sqrt3\,\nu B^3
 +O(\lambda^2B+\lambda B^3+B^5)+O(B^{11}).
\end{equation}
At the nonzero branch,
$B_\lambda^2=\lambda/(8\sqrt3\,\nu)+O(\lambda^2)$.  Differentiating \eqref{eq:reflection-amplitude-dynamics} at $B=B_\lambda$ gives
\[
 \partial_B\dot B\big|_{B_\lambda}
 =\frac\lambda4-6\sqrt3\,\nu B_\lambda^2+O(\lambda^2)
 =-\frac\lambda2+O(\lambda^2),
\]
which is negative for sufficiently small $\lambda$.  The range spectrum remains in
$\{\operatorname{Re}z\le-g_0/2\}$ by finite-dimensional spectral continuity and the analytic range correction.  The standard linearized-stability theorem for a smooth finite-dimensional ODE then gives local asymptotic stability.  No claim is made here about stability in the full period-cell space, where the commensurate phase-locking eigenvalue is of order $\lambda^5$ and its sign depends on the locking coefficient.
\end{proof}

\begin{remark}
Since $k_* = \pi/6$, the lattice branch has period $12$ and only a discrete translation orbit.  The site- and bond-centered branches correspond to the two reflection fixed spaces constructed in \Cref{thm:branch}; the bond-centered branch $\phi_0=\pi/12$ is the one used in the numerical section.
\end{remark}

\subsection{Lindblad realization}
In this subsection we write
\[
 a_x:=a_{x,M_N},\qquad Q_x:=Q_{x,N},\qquad P_x:=P_{x,N},
\]
and suppress the cutoff subscripts to keep formulas readable.  Let $E_j$ be the nearest-neighbor bonds in direction $e_j$, and set
\[
  K_j=\frac{D_{q,j}+D_{p,j}}2,
  \qquad
  K_j'=\frac{D_{q,j}-D_{p,j}}2.
\]
For $h\ge0$ and $\mu\in C_{12}$, define the optional coherent phase seed

\[
H_{\rm seed}^{h,\mu}=i\sqrt N\sum_x(h_x^{(\mu)}a_x^*-\overline{h_x^{(\mu)}}a_x),
\qquad
h_x^{(\mu)}=\frac{h}{\sqrt2}(r_q+ir_p)\cos(k_*\cdot x+\phi_0+\mu\pi/6).
\]

The unseeded construction is obtained by setting $h=0$; the optional seed is only a finite-volume phase selector for a discrete translate.  The direction-dependent construction is
\begin{align*}
H_{\rm loc}&=\sum_x\left[\Omega_\lambda a_x^*a_x+\frac{i\varepsilon}{2}((a_x^*)^2-a_x^2)\right],\\
H_{\rm sq}&=-\frac i2\sum_{j=1}^dK_j'
\sum_{\langle x,y\rangle\in E_j}
\left[((a_x^*-a_y^*)^2)-(a_x-a_y)^2\right],\\
\cL^*(\rho)&=-i[H_{\rm loc}+H_{\rm sq}+H_{\rm seed}^{h,\mu},\rho]
+\kappa\sum_x\cD[a_x]\rho+\frac{\gamma}{N}\sum_x\cD[a_x^2]\rho\\
&\quad+2\sum_{j=1}^dK_j\sum_{\langle x,y\rangle\in E_j}\cD[a_x-a_y]\rho.
\end{align*}

\begin{theorem}\label[theorem]{thm:lindblad}
With
\[
\varepsilon=\frac{a+b}{2},\qquad \kappa=b-a,
\qquad K_j=\frac{D_{q,j}+D_{p,j}}2,
\qquad K_j'=\frac{D_{q,j}-D_{p,j}}2,
\qquad \gamma=2\nu,
\]
the unseeded generator realizes every equation of the form \eqref{eq:rtclass} in the local coherent-state topology, with an $O_{\cS_{L,N}(E)}(N^{-1})$ remainder uniformly on bounded coherent-state sets.  Adding $H_{\rm seed}^{h,\mu}$ adds precisely the deterministic forcing $h r\cos(k_*\cdot x+\phi_0+\mu\pi/6)$ to the coherent first-moment equation.  The stripe and isotropic specializations are obtained from the transport choices stated after \eqref{eq:rtclass}.
\end{theorem}
\begin{proof}
The local Hamiltonian and one-photon loss give
\[
\dot q=(\varepsilon-\kappa/2)q+\Omega_\lambda p,
\qquad
\dot p=-\Omega_\lambda q-(\varepsilon+\kappa/2)p.
\]
For each $j$, dissipative hopping on $E_j$ contributes $K_j\Delta_j$ to both quadratures, while the squeezing-hopping Hamiltonian contributes $K_j'\Delta_j$ to $q$ and $-K_j'\Delta_j$ to $p$.  Finally,
\[
\frac{\gamma}{N}\mathcal D[a_x^2]^*(a_x)=-\frac{\gamma}{N}a_x^*a_x^2.
\]
For
\[
\alpha_x=N^{-1/2}\langle a_x\rangle=\frac{q_x+ip_x}{\sqrt2},
\]
this gives $\dot\alpha_x=-\gamma|\alpha_x|^2\alpha_x+O(N^{-1})$.  The real cubic coefficient is $\gamma/2$, and $\gamma=2\nu$ realizes the coefficient in \eqref{eq:q}--\eqref{eq:p}.  Combining the local terms gives
\begin{align*}
\cL Q_x&=aQ_x+\Omega_\lambda P_x-\nu(Q_x^2+P_x^2)Q_x+\sum_jD_{q,j}\Delta_jQ_x+O_{\cS_{L,N}(E)}(N^{-1}),\\
\cL P_x&=-\Omega_\lambda Q_x-bP_x-\nu(Q_x^2+P_x^2)P_x+\sum_jD_{p,j}\Delta_jP_x+O_{\cS_{L,N}(E)}(N^{-1}).
\end{align*}
All ordering errors are $O_{\cS_{L,N}(E)}(N^{-1})$ by \Cref{lem:symbol}; the cutoff-tail contribution is exponentially small by \Cref{lem:comm} and is absorbed in that coherent-state remainder.
\end{proof}

\begin{remark}\label[remark]{rem:multiphoton-saturation}
The Lindblad construction also realizes positive radial odd-polynomial saturation.  For a multiphoton order $r_{\rm ph}\ge2$, the local channel
\[
 \gamma_{r_{\rm ph}} N^{1-r_{\rm ph}}\cD[a_x^{r_{\rm ph}}]
\]
contributes in the coherent limit
\[
 \dot\alpha_x=-\frac{r_{\rm ph}\gamma_{r_{\rm ph}}}{2}|\alpha_x|^{2r_{\rm ph}-2}\alpha_x,
\]
and hence a negative radial drift proportional to
$-(q_x^2+p_x^2)^{r_{\rm ph}-1}(q_x,p_x)$.  Finite sums of such channels therefore cover positive radial odd-polynomial reaction terms; the cubic model used below is the case $r_{\rm ph}=2$.
\end{remark}

\begin{lemma}\label[lemma]{lem:bragg}
Let
\[
  u_{\lambda,\mu}(x)=2B_\lambda r\cos(k_*\cdot x+\phi_0+\mu\pi/6)+O(\lambda^{3/2}),
  \qquad \mu\in C_{12},
\]
be the commensurate stripe branch, and let $\mathcal O_x^{\rm pat}=v_{\rm pat}\cdot O_{x,N}$ with $v_{\rm pat}\cdot r\ne0$.  The discrete phase-averaged weight satisfies
\[
  \frac{S_\lambda^{\rm coh}(k_*)}{|\Lambda_L|}
  =B_\lambda^2|v_{\rm pat}\cdot r|^2+O(\lambda^2)>0
\]
for all sufficiently small $\lambda>0$ along period-compatible tori.
\end{lemma}
\begin{proof}
The $+k_*$ Fourier coefficient of $2B_\lambda\cos(k_*\cdot x+\phi_0+\mu\pi/6)$ is $B_\lambda e^{i(\phi_0+\mu\pi/6)}$, so the normalization of $S(k_*)$ gives the coefficient $B_\lambda^2$ rather than $2B_\lambda^2$.  Since $B_\lambda=O(\lambda^{1/2})$, the cross term between the leading branch and the $O(\lambda^{3/2})$ branch correction contributes $O(\lambda^2)$ to the Bragg weight.  Averaging over the twelve lattice translates gives
\[
 \frac{1}{12} \sum_{\mu=0}^{11}\cos(k_*\cdot x+\phi_0+\mu\pi/6)\cos(k_*\cdot y+\phi_0+\mu\pi/6)
 =\frac{1}{2}\cos(k_*\cdot(x-y)),
\]
because the discrete sum of $e^{2i\mu\pi/6}$ vanishes.  The discrete phase-averaged two-point function has leading term $2B_\lambda^2|v_{\rm pat}\cdot r|^2\cos(k_*\cdot(x-y))$.  Since $2k_*\not\equiv0\pmod{2\pi}$, Fourier orthogonality eliminates the nonzero $2k_*$ character.  Substitution into the structure factor at $k_*$ gives the displayed formula.  This is the commensurate finite-orbit specialization of \Cref{lem:general-bragg}.
\end{proof}

The local pattern field is $\mathcal O_x^{\rm pat}=Q_x-\beta P_x$.  The local soft field is $\mathcal O_x^{\rm soft}=Q_x+\beta P_x$; its Fourier component $\widehat{\mathcal O}^{\rm soft}(k_*)$
spans the left zero-eigenvalue mode of the homogeneous
first-moment linearized Heisenberg generator at $k_*$.

\section{Pattern formation and quantum fluctuations}\label{sec:patterns-fluctuations}
The first-moment dynamics displays three pattern classes.  A weakly seeded stripe follows the analytic branch, while a broadband-noise run selects the same finite-wave-number band without a critical-mode seed.  With isotropic transport, the instability produces spot and labyrinth regimes on a finite-$k$ shell.  Gaussian fluctuations are obtained by displacing the microscopic generator about the stationary stripe or along the two isotropic trajectories.

\subsection{Numerical patterns and convergence}
The simulations use periodic boundary conditions and the unitary discrete Fourier convention stated before \Cref{def:turing-point}.  For the chosen parameter regime, let $A_{\rm lin}(k):=\widehat A_\lambda(k)$ be the linear lattice symbol and let $G(u)_x=-\nu|u_x|^2u_x$ be the onsite nonlinear remainder.  The IMEX Euler update is
\[
  \widehat u^{\,n+1}(k)
  =\bigl[I-\Delta t_{\rm PDE}\,A_{\rm lin}(k)\bigr]^{-1}
   \left[\widehat u^{\,n}(k)
   +\Delta t_{\rm PDE}\,\widehat{G(u^n)}(k)\right].
\]
The baseline coherent-field step is $\Delta t_{\rm PDE}:=0.05$.  The stripe calculation uses
\[
  (\lambda,\nu,D_y,L,T)=(0.4,4,0.2,192,160).
\]
The isotropic spot and labyrinth calculations use
\[
  (a,b,\Omega_\lambda,D_q,D_p,\nu,L)=(1,3,1.8,0.6,4.5,4,128),
\]
with final times $T=50$ and $T=80$.  Localized nuclei and weak noise select the Spot basin, while broadband noise produces the Labyrinth regime; in both cases the unstable band selects the wavelength.

The displayed pattern field is
\[
  \mathcal O_{\rm stripe}^{\rm pat}=q-\beta p,
  \qquad
  \mathcal O_{\rm iso}^{\rm pat}=q-0.40p.
\]
For the isotropic parameters, the maximal linear growth rate occurs at $k\simeq0.649$, where the right eigenvector normalized by its $q$ component is
\[
  r_{\max}=(1,-0.365\ldots)^{\mathsf T}.
\]
$v_{\rm iso}=(1,-0.40)$ is a fixed visualization covector with nonzero overlap with the unstable eigenvector and is used unchanged for all isotropic runs.
For either field, define its spatial mean by
$\overline{\mathcal O}^{\rm pat}:=|\Lambda_L|^{-1}\sum_{x\in\Lambda_L}\mathcal O_x^{\rm pat}$ and set
\[
  \widehat{\mathcal O}^{\rm pat}(k)
  :=|\Lambda_L|^{-1/2}\sum_{x\in\Lambda_L}
    \bigl(\mathcal O_x^{\rm pat}-\overline{\mathcal O}^{\rm pat}\bigr)e^{-ik\cdot x},
  \qquad k_j=\frac{2\pi n_j}{L},\quad n_j\in\mathbb Z/L\mathbb Z.
\]
Radial bins are the half-open annuli of width $\Delta k=2\pi/L$ centered at integer multiples of $\Delta k$; $\overline P(\varrho)$ is the arithmetic mean of $|\widehat{\mathcal O}^{\rm pat}(k)|^2$ in the annulus centered at the radial coordinate $\varrho$.  Define
\[
 k_{\rm dom}=\underset{\varrho>0}{\operatorname{arg\,max}}\,\overline P(\varrho),
 \qquad
 C_{\rm shell}=\frac{\sum_{||k|-k_{\rm dom}|\le\Delta k}|\widehat{\mathcal O}^{\rm pat}(k)|^2}
 {\sum_{k\ne0}|\widehat{\mathcal O}^{\rm pat}(k)|^2},
\]
with ties resolved toward the smaller radius.  Taking one Fourier spacing as the radial half-width makes $C_{\rm shell}$ resolution-normalized across grids.  The reported stationary residual is
\[
 r_{\rm stat}:=\frac{\|\mathcal F_\lambda^{\Lambda_L}(u_{\rm final})\|_2}
 {\|u_{\rm final}\|_2},
 \qquad u_{\rm final}=(q_{\rm final},p_{\rm final}),
\]
where $\mathcal F_\lambda^{\Lambda_L}$ is the full lattice drift in \eqref{eq:closure}, specialized to the simulated regime.  Refinement and control calculations preserve the selected finite-wave-number band.  The stripe stationary residual is below $8.6\times10^{-4}$ in the reported controls, and the Newton-continued period-cell profiles have infinity-norm residual below $10^{-8}$.

At $\lambda=0$ the maximal real eigenvalue of $\widehat A_\lambda(k_*)$ touches zero while the homogeneous mode remains stable; its derivative with respect to $\lambda$ is $1/4$.

\begin{figure}[H]
\centering
\includegraphics[width=\textwidth,height=0.62\textheight,keepaspectratio]{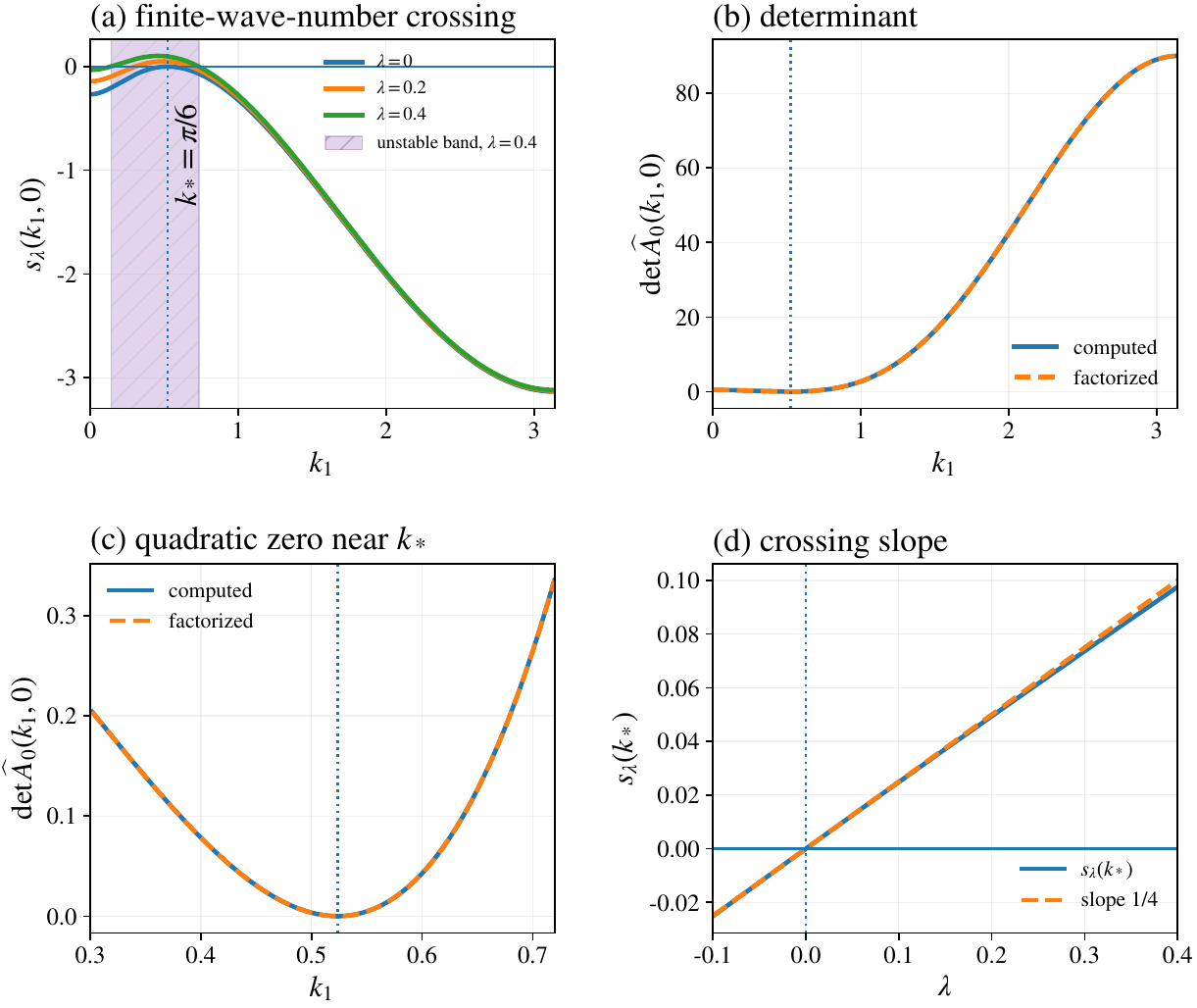}
\caption[Turing point]{Finite wave-number Turing point.  (a) The critical eigenvalue crosses zero at $k_*=\pi/6$ while the homogeneous mode remains stable; the shaded interval is the unstable band at $\lambda=0.4$.  (b,c) The computed determinant agrees with the exact factorization and has a quadratic zero at $k_*$.  (d) The critical eigenvalue crosses with slope $1/4$.}
\label{fig:turing-point}
\end{figure}

The period-compatible grid represents $k_*=\pi/6$ without wave-number discretization error, and the measured $+k_*$ Fourier coefficient agrees with the leading branch prediction to about $0.2\%$.

\begin{figure}[H]
\centering
\includegraphics[width=\textwidth,height=0.62\textheight,keepaspectratio]{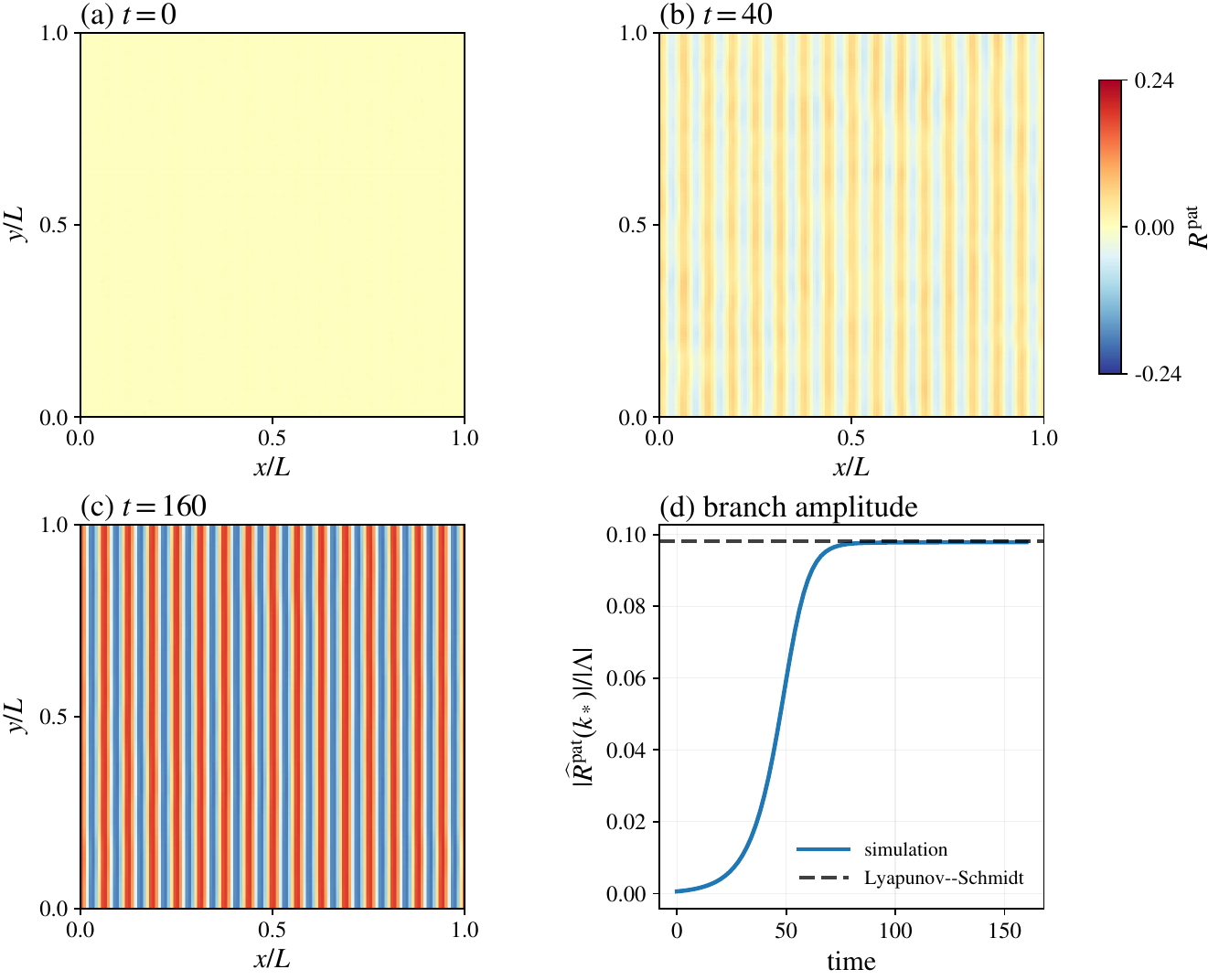}
\caption[Evolution toward the commensurate stripe]{Evolution toward the commensurate stripe.  (a--c) First-moment field at $t=0$, $40$, and $160$.  (d) The measured $+k_*$ Fourier amplitude approaches the Lyapunov--Schmidt prediction.}
\label{fig:stripe}
\end{figure}

Across $0.05\le\lambda\le0.4$, the continued branch agrees with the leading Lyapunov--Schmidt amplitude law near threshold, the reflection-fixed spectral abscissa remains negative, and the period-cell NPT minimum remains at $k_*/\pi=1/6$.
\FloatBarrier

\begin{table}[htbp]
\centering
\small
\setlength{\tabcolsep}{7pt}
\begin{tabular}{@{}lccc@{}}
\toprule
Regime & $L$ & $k_{\rm dom}$ & $C_{\rm shell}$\\
\midrule
Stripe & 192 & 0.52 & $>0.99$\\
Spot & 128 & 0.59 & 0.56\\
Labyrinth & 128 & 0.64 & 0.57\\
\bottomrule
\end{tabular}
\caption[Dominant wave number and shell concentration]{Dominant wave number and shell concentration.  The shell concentration uses one Fourier spacing as its radial half-width, as defined above.}
\label{tab:patterns}
\end{table}

The dominant power lies inside the corresponding linear unstable band in all three regimes.

\subsection{Gaussian fluctuation spectra and mode entanglement}
Let $\mathcal C$ be a finite period cell, with addition taken periodically, and define its directional bond set and periodic difference operators by
\[
  E_j(\mathcal C):=\bigl\{\{x,x+e_j\}:x\in\mathcal C\bigr\},\qquad
  (\Delta_j^{\mathcal C}f)_x:=f_{x+e_j}+f_{x-e_j}-2f_x,
  \qquad \mathsf L_j^{\mathcal C}:=-\Delta_j^{\mathcal C}.
\]
Let
\[
  u^\star=(q_x^\star,p_x^\star)_{x\in\mathcal C},
  \qquad
  \alpha_x^\star=\frac{q_x^\star+ip_x^\star}{\sqrt2},
\]
be a stationary periodic branch.  Linearization of the local Lindblad generator produces a Gaussian dynamics on $\mathcal C$.  When its drift is Hurwitz, the stationary covariance is determined by an algebraic Lyapunov equation.  The finite-time theorem below controls the approximation by the exact microscopic dynamics, and the homogeneous case is then solved in closed form.

\begin{proposition}\label[proposition]{prop:branch-gaussian}
Displace the microscopic modes according to
\[
  a_x=\sqrt N\,\alpha_x^\star+b_x.
\]
After cancellation of the $O(\sqrt N)$ terms by the stationary branch equation, the $O(N^0)$ fluctuation generator is
\begin{align}
\cL_{\star}^{(2)*}(\rho)
={}&-i[H_\star^{(2)},\rho]
+\kappa\sum_{x\in\mathcal C}\mathcal D[b_x]\rho
+\sum_{x\in\mathcal C}\mathcal D[2\sqrt\gamma\,\alpha_x^\star b_x]\rho\notag\\
&\quad+2\sum_{j=1}^dK_j
\sum_{\langle x,y\rangle\in E_j(\mathcal C)}\mathcal D[b_x-b_y]\rho,\label{eq:branch-gaussian-generator}
\end{align}
where
\begin{align}
H_\star^{(2)}
={}&\sum_{x\in\mathcal C}\left[
\Omega_\lambda b_x^*b_x
+\frac i2\left\{
\bigl(\varepsilon-\gamma(\alpha_x^\star)^2\bigr)b_x^{*2}
-\bigl(\varepsilon-\gamma(\overline{\alpha_x^\star})^2\bigr)b_x^2
\right\}\right]\notag\\
&-\frac i2\sum_{j=1}^dK_j'
\sum_{\langle x,y\rangle\in E_j(\mathcal C)}
\left[(b_x^*-b_y^*)^2-(b_x-b_y)^2\right].\label{eq:branch-gaussian-hamiltonian}
\end{align}
Define the canonical fluctuation quadratures by
\[
  \delta q_x:=\frac{b_x+b_x^*}{\sqrt2},\qquad
  \delta p_x:=\frac{b_x-b_x^*}{i\sqrt2},
\]
and collect them in
\[
  \xi:=(\delta q_1,\delta p_1,\ldots,\delta q_{|\mathcal C|},\delta p_{|\mathcal C|})^{\mathsf T}.
\]
Here and below $\langle\cdot\rangle$ denotes expectation in the Gaussian fluctuation state under discussion.  The first-moment drift is
\[
  \frac d{dt}\langle\xi\rangle=A_{\rm per}\langle\xi\rangle,
  \qquad
  A_{\rm per}:=D_u\mathcal F_\lambda^{\mathcal C}(u^\star),
\]
where the site-diagonal and nearest-neighbor blocks are
\begin{align}
(A_{\rm per})_{xx}
={}&\begin{pmatrix}
 a-\nu(3q_x^{\star2}+p_x^{\star2}) & \Omega_\lambda-2\nu q_x^\star p_x^\star\\
 -\Omega_\lambda-2\nu q_x^\star p_x^\star & -b-\nu(q_x^{\star2}+3p_x^{\star2})
\end{pmatrix}
-2\sum_{j=1}^d\begin{pmatrix}D_{q,j}&0\\0&D_{p,j}\end{pmatrix},\label{eq:Aper-diagonal}\\
(A_{\rm per})_{x,x\pm e_j}
={}&\begin{pmatrix}D_{q,j}&0\\0&D_{p,j}\end{pmatrix},\label{eq:Aper-neighbor}
\end{align}
with periodic wrap on $\mathcal C$.  If $\mathsf L_j^{\mathcal C}=-\Delta_j^{\mathcal C}$ is the period-cell graph Laplacian, the symmetrized covariance
\[
  V_{mn}=\frac12\langle\xi_m\xi_n+\xi_n\xi_m\rangle
  -\langle\xi_m\rangle\langle\xi_n\rangle
\]
satisfies
\begin{equation}\label{eq:period-cell-lyapunov}
  \dot V=A_{\rm per}V+VA_{\rm per}^{\mathsf T}+D_{\rm per},
  \qquad
  D_{\rm per}=
  \bigoplus_{x\in\mathcal C}\left[\frac\kappa2+\gamma(q_x^{\star2}+p_x^{\star2})\right]I_2
  +\sum_{j=1}^dK_j\,\mathsf L_j^{\mathcal C}\otimes I_2.
\end{equation}
If $A_{\rm per}$ is Hurwitz, the stationary covariance is unique and is given by
\[
  V_{\rm per}=\int_0^\infty e^{tA_{\rm per}}D_{\rm per}e^{tA_{\rm per}^{\mathsf T}}\,dt,
\]
or equivalently by the algebraic Lyapunov equation
\[
  A_{\rm per}V_{\rm per}+V_{\rm per}A_{\rm per}^{\mathsf T}+D_{\rm per}=0.
\]
\end{proposition}
\begin{proof}
The shifted two-photon jump expands as
\[
  \sqrt{\gamma/N}\,a_x^2
  =\sqrt{\gamma N}(\alpha_x^\star)^2
   +2\sqrt\gamma\,\alpha_x^\star b_x
   +\sqrt{\gamma/N}\,b_x^2.
\]
For a scalar $c$ and an operator $X$,
\[
  \mathcal D[c\mathbf1+X]\rho
  =\mathcal D[X]\rho
  -i\left[\frac i2(\overline cX-cX^*),\rho\right].
\]
The constant--linear cross terms, together with the displaced Hamiltonian and bond terms, are precisely the $O(\sqrt N)$ stationary first-moment equation and vanish at $u^\star$.  The constant--quadratic cross term produces
\[
  -\frac{i\gamma}{2}\left[(\alpha_x^\star)^2b_x^{*2}
  -(\overline{\alpha_x^\star})^2b_x^2\right],
\]
while $2\sqrt\gamma\,\alpha_x^\star b_x$ is the surviving linear jump.  The remaining quadratic-jump and linear--quadratic cross terms are lower order in $N^{-1/2}$.  The other shifted jumps retain the linear parts displayed in \eqref{eq:branch-gaussian-generator}, and the quadratic Hamiltonians retain the form in \eqref{eq:branch-gaussian-hamiltonian}.

The first-moment equation of this Gaussian generator is the linearization of the semiclassical vector field, giving \eqref{eq:Aper-diagonal}--\eqref{eq:Aper-neighbor}.  One-photon loss contributes $\kappa I_2/2$ to the local diffusion, the linearized two-photon jump contributes $\gamma(q_x^{\star2}+p_x^{\star2})I_2$, and the jump $\sqrt{2K_j}(b_x-b_y)$ contributes $K_j$ times the corresponding edge Laplacian in each quadrature.  These contributions give $D_{\rm per}$ in \eqref{eq:period-cell-lyapunov}.  The standard moment equations for a Gaussian Lindblad generator then yield the Lyapunov equation and its integral solution when $A_{\rm per}$ is Hurwitz.
\end{proof}

For comparison with the microscopic dynamics, let $\Lambda$ be a fixed torus and let
\[
  u(t)=(q_x(t),p_x(t))_{x\in\Lambda},
  \qquad 0\le t\le T,
\]
be a bounded solution of the first-moment equation.  Write
$\alpha_x(t)=(q_x(t)+ip_x(t))/\sqrt2$.  For a one-site amplitude $\alpha$, define the normalized projected coherent state
\[
  |\psi_\alpha^{(N)}\rangle
  :=\frac{\mathbf P_{\le M_N}|\sqrt N\alpha\rangle}
  {\|\mathbf P_{\le M_N}|\sqrt N\alpha\rangle\|}.
\]
Let $\rho_{\Lambda,N}(t)$ be the unseeded finite-volume Lindblad evolution on $\Lambda$ with initial state
\begin{equation}\label{eq:trajectory-projected-coherent-initial-data}
  \rho_{\Lambda,N}(0)
  =\bigotimes_{x\in\Lambda}
  |\psi_{\alpha_x(0)}^{(N)}\rangle
  \langle\psi_{\alpha_x(0)}^{(N)}|.
\end{equation}
On one full-Fock mode we use the Weyl convention
\[
  W(z):=\exp(za^*-\overline z a),\qquad W(z)^*aW(z)=a+z.
\]
Let $\iota_N$ be the canonical embedding of the cutoff tensor product into the full bosonic Fock space, with operators and states extended by zero on its orthogonal complement.  Set
\[
  W_N(t):=\bigotimes_{x\in\Lambda}W(\sqrt N\alpha_x(t)),
  \qquad
  \widetilde\rho_{\Lambda,N}(t)
  :=W_N(t)^*\iota_N(\rho_{\Lambda,N}(t))W_N(t).
\]
Thus
\[
  W_N(t)^*(a_x-\sqrt N\alpha_x(t))W_N(t)=b_x,
  \qquad
  \mathfrak N_b:=\sum_{x\in\Lambda}b_x^*b_x,
\]
where the $b_x$ are canonical full-Fock annihilation operators.  All identities involving the unbounded full-Fock operators $b_x,b_x^*$ and their polynomial Lindblad generators are first understood as quadratic-form identities on the finite-particle core
\[
  \mathcal D_{\rm fin}
  =\operatorname{span}_{\rm alg}\left\{
  |n_1,\ldots,n_{|\Lambda|}\rangle:
  n_x\in\{0,1,2,\ldots\}
  \right\}.
\]
The moment bounds below justify taking expectations of these identities and extending them to the evolved states considered here.  Up to a real scalar phase $\vartheta_N(t)I$, the derivative of the time-dependent Weyl displacement is
\[
  iW_N(t)^*\dot W_N(t)
  =i\sqrt N\sum_{x\in\Lambda}
  \left[\dot\alpha_x(t)b_x^*
  -\dot{\overline{\alpha_x(t)}}b_x\right]
  +\vartheta_N(t)I.
\]
Let $H_{\Lambda,N}$ denote the Hamiltonian part of the finite-volume generator on $\Lambda$, and let $H_{\Lambda,N}^{\rm ext}:=\iota_N(H_{\Lambda,N})$ be its zero extension.  The Hamiltonian in the moving frame is
\begin{equation}\label{eq:moving-frame-hamiltonian}
  H_{\rm mov}(t):=W_N(t)^*H_{\Lambda,N}^{\rm ext}W_N(t)-iW_N(t)^*\dot W_N(t).
\end{equation}
The moving-frame Heisenberg generator $\widetilde\cL_{\Lambda,N}^{\rm fl}(t)$ is defined on fluctuation polynomials by
\begin{equation}\label{eq:moving-frame-generator-definition}
  \frac d{dt}\widetilde\rho_{\Lambda,N}(t)(A)
  =\widetilde\rho_{\Lambda,N}(t)
  \bigl(\widetilde\cL_{\Lambda,N}^{\rm fl}(t)A\bigr).
\end{equation}
Let
$\mathcal F_{\lambda,\mathbb C,x}^{\Lambda}(u)
:=\bigl(\mathcal F_{\lambda,q,x}^{\Lambda}(u)+i\mathcal F_{\lambda,p,x}^{\Lambda}(u)\bigr)/\sqrt2$
be the complex form of the full first-moment drift.  The total $O(\sqrt N)$ linear Hamiltonian is therefore proportional to
\[
  i\sqrt N\sum_x
  \left[
  (\mathcal F_{\lambda,\mathbb C,x}^{\Lambda}(u(t))-\dot\alpha_x)b_x^*
  -(\overline{\mathcal F_{\lambda,\mathbb C,x}^{\Lambda}(u(t))}-\dot{\overline\alpha_x})b_x
  \right].
\]
This linear term vanishes because $u(t)$ solves the first-moment equation.  We denote by $\cL_{u(t)}^{(2)}$ the surviving quadratic Heisenberg generator: its drift is $D_u\mathcal F_\lambda^{\Lambda}(u(t))$ and its diffusion is assembled from the linear parts of the displaced jumps.  For a stationary branch $u(t)=u^\star$, the displacement $W_N(t)$ is constant, and $\cL_{u^\star}^{(2)}$ is the Heisenberg adjoint of the Schr\"odinger-picture generator $\cL_\star^{(2)*}$ in \Cref{prop:branch-gaussian}.

To isolate the cutoff boundary, write the zero-extended microscopic Heisenberg generator as a sum over local supports $Z\subset\Lambda$.  Let $\mathfrak P_Z^{\rm tr}$ be the zero extension of the cutoff local superoperator supported on $Z$, and let $\mathfrak P_Z^{\rm full}$ be the same polynomial superoperator evaluated with canonical full-Fock operators.  Their moving-frame discrepancy is, by definition,
\begin{equation}\label{eq:boundary-remainder-definition}
  \mathcal E_{\Lambda,N}^{\rm bdry}(t)A
  :=W_N(t)^*\sum_Z\bigl(\mathfrak P_Z^{\rm tr}-\mathfrak P_Z^{\rm full}\bigr)
  \bigl(W_N(t)AW_N(t)^*\bigr)W_N(t).
\end{equation}
Below, $\mathfrak P^{\rm tr}:=\sum_Z\mathfrak P_Z^{\rm tr}$ and $\mathfrak P^{\rm full}:=\sum_Z\mathfrak P_Z^{\rm full}$.  Within an estimate at a fixed time, the shorthand $\widetilde\rho$ means $\widetilde\rho_{\Lambda,N}(t)$.

\begin{lemma}\label[lemma]{lem:displaced-boundary-remainder}
Fix $\Lambda$, $T<\infty$, and a bounded coherent trajectory $u(t)$.  For the projected coherent initial data \eqref{eq:trajectory-projected-coherent-initial-data}, choose $c\ge c_T$ as in \Cref{lem:cutoff-tail-propagation}.  Let
\[
 A_0=\mathscr P(b,b^*),
 \qquad G_t=W_N(t)A_0W_N(t)^*
\]
for a fixed-degree fluctuation polynomial $\mathscr P$.  Then there are $C_{A_0,\Lambda,T}<\infty$ and $c_{A_0}>0$, independent of $N$ and of all larger cutoff ratios $c$, such that
\[
 \sup_{0\le t\le T}
 \left|\widetilde\rho_{\Lambda,N}(t)
 \left(\mathcal E_{\Lambda,N}^{\rm bdry}(t)A_0\right)\right|
 \le C_{A_0,\Lambda,T}e^{-c_{A_0}N}.
\]
\end{lemma}
\begin{proof}
By the definition of the moving frame,
\begin{equation}\label{eq:boundary-remainder-original-frame}
 \widetilde\rho_{\Lambda,N}(t)
 \left(\mathcal E_{\Lambda,N}^{\rm bdry}(t)A_0\right)
 =\rho_{\Lambda,N}(t)
 \left((\mathfrak P^{\rm tr}-\mathfrak P^{\rm full})(G_t)\right).
\end{equation}
Here $\mathfrak P^{\rm tr}$ is the fixed-degree local Heisenberg superoperator built from projected creation and annihilation operators, and $\mathfrak P^{\rm full}$ is its zero-extended full-Fock counterpart.  Let $\mathbf P_{\Lambda,M_N}$ denote the zero-extension projection onto the cutoff-supported tensor product.  For some $m_{A_0}\in\mathbb N$ and a fixed buffer $C_{\rm buf}$ depending only on the polynomial degree and interaction range, telescoping each operator word gives the cutoff-supported identity
\begin{equation}\label{eq:boundary-superoperator-factorization}
 \mathbf P_{\Lambda,M_N}
 \Bigl[(\mathfrak P^{\rm tr}-\mathfrak P^{\rm full})(G_t)\Bigr]
 \mathbf P_{\Lambda,M_N}
 =\sum_{r=1}^{m_{A_0}}A_r(t)\Pi_{x_r,N}^{\rm bdry}B_r(t),
 \qquad
 \Pi_{x,N}^{\rm bdry}=\mathbf1_{\{n_x\ge M_N-C_{\rm buf}\}}.
\end{equation}
Both the input and output in \eqref{eq:boundary-superoperator-factorization} lie in the cutoff-supported subspace.  A word of degree $d$ can cross the cutoff only within $d$ occupation levels of the top state, and
$X_1\cdots X_m-Y_1\cdots Y_m=\sum_jX_1\cdots X_{j-1}(X_j-Y_j)Y_{j+1}\cdots Y_m$
places one cutoff defect in every summand.  The commutator defect
$[a_{M_N},a_{M_N}^*]-I=-(M_N+1)|M_N\rangle\langle M_N|$
is the simplest example.

Dissipative terms have the same support structure.  Writing $L_M=L+\delta L$ after zero extension,
\begin{align*}
 L_M^*G_tL_M-L^*G_tL
 &=\delta L^*G_tL_M+L^*G_t\delta L,\\
 L_M^*L_M-L^*L
 &=\delta L^*L_M+L^*\delta L,
\end{align*}
and every $\delta L$ is supported on the original-space boundary projector.  Hamiltonian commutators, jump cross terms, and anticommutators therefore have the form \eqref{eq:boundary-superoperator-factorization}.  Each $A_r(t)$ and $B_r(t)$ is a finite sum of cutoff-restricted words of fixed degree.  The degree bound depends only on the degree of $A_0$ and the local generator, while the coefficients are polynomials in the bounded trajectory amplitudes and in $\sqrt N$.

Weyl conjugation preserves polynomial degree but does not leave the coefficients $N$-independent: each displacement contributes a factor $\sqrt N\,\alpha_x(t)$.  If $d_{A_0}$ is the field degree of $A_0$, boundedness of the trajectory and
\[
 \norm{a_x\mathbf P_{\Lambda,M_N}},\ \norm{a_x^*\mathbf P_{\Lambda,M_N}}
 \le \sqrt{M_N+1}
\]
give
\begin{equation}\label{eq:weyl-conjugated-polynomial-bound}
 \norm{\mathbf P_{\Lambda,M_N}G_t\mathbf P_{\Lambda,M_N}}
 \le C_{A_0,T}\bigl(\sqrt N+\sqrt{M_N+1}\bigr)^{d_{A_0}}
 \le C_{A_0,T}(cN)^{d_{A_0}/2}.
\end{equation}
After the fixed-degree factors $A_r(t)$ and $B_r(t)$ in \eqref{eq:boundary-superoperator-factorization} are included, Cauchy--Schwarz produces at most a polynomial factor $C_{A_0,\Lambda,T}(cN)^{p_{A_0}}$ for some finite $p_{A_0}$.  Using \eqref{eq:strong-cutoff-boundary-tail}, we obtain
\[
 C_{A_0,\Lambda,T}(cN)^{p_{A_0}}\rho_{\Lambda,N}(t)(\Pi_{x,N}^{\rm bdry})^{1/2}
 \le C_{A_0,\Lambda,T}(cN)^{p_{A_0}}
 \exp\left\{-\frac N2[\theta c-m_T(\theta,\zeta)]\right\}.
\]
Put $g(c)=\theta c-m_T(\theta,\zeta)$ and $\delta=g(c_T)>0$.  For $c\ge c_T$,
$cN\le C_{\theta,\delta}Ng(c)$, so
\[
 (cN)^{p_{A_0}}e^{-Ng(c)/2}
 \le C(Ng(c))^{p_{A_0}}e^{-Ng(c)/4}e^{-\delta N/4}
 \le C'e^{-\delta N/4}.
\]
Summation over the fixed torus proves the estimate uniformly over all larger cutoff ratios.
\end{proof}

\begin{lemma}\label[lemma]{lem:moving-frame-remainders}
On every fixed torus and bounded coherent trajectory, define
\[
 L_x^{(1)}(t)=2\sqrt\gamma\,\alpha_x(t)b_x,
 \qquad L_x^{(2)}=\sqrt\gamma\,b_x^2,
\]
and, for two jump operators $L,Q$, set
\begin{equation}\label{eq:cross-dissipator}
 \mathcal B_{L,Q}^*(A)
 =L^*AQ+Q^*AL-\frac12\{L^*Q+Q^*L,A\}.
\end{equation}
Then the shifted generator on fixed-degree fluctuation polynomials has the exact decomposition
\begin{equation}\label{eq:exact-shifted-generator-expansion}
 \widetilde\cL_{\Lambda,N}^{\rm fl}(t)
 =\cL_{u(t)}^{(2)}
 +N^{-1/2}\sum_{x\in\Lambda}\mathcal B_{L_x^{(1)}(t),L_x^{(2)}}^*
 +\frac\gamma N\sum_{x\in\Lambda}\cD[b_x^2]^*
 +\mathcal E_{\Lambda,N}^{\rm bdry}(t).
\end{equation}
The cross superoperator maps a quadratic fluctuation observable to a fluctuation polynomial of degree at most three, while the residual dissipator maps it to one of degree at most four.  Their coefficients are uniform in $N$ on $[0,T]$.  For $G=(1+\mathfrak N_b)^2$,
\begin{equation}\label{eq:remainder-number-bounds}
 \left|\widetilde\rho\!\left(\sum_x\mathcal B_{L_x^{(1)},L_x^{(2)}}^*G\right)\right|
 \le C\widetilde\rho((1+\mathfrak N_b)^{5/2}),
\end{equation}
and the quartic term has the exact sign
\begin{equation}\label{eq:residual-two-photon-coercivity}
 \gamma\sum_x\cD[b_x^2]^*G
 =-4\gamma\mathfrak N_b\sum_xb_x^{*2}b_x^2.
\end{equation}
There is no additional quartic remainder away from the cutoff boundary.
\end{lemma}
\begin{proof}
Substitute $a_x=\sqrt N\alpha_x(t)+b_x$ in every local Hamiltonian and jump, and include the derivative term in \eqref{eq:moving-frame-hamiltonian}.  All Hamiltonians are at most quadratic and every jump other than the two-photon jump is linear, so after the coherent equation cancels the $O(\sqrt N)$ terms they contribute only to the quadratic generator $\cL_{u(t)}^{(2)}$ and to the cutoff-boundary error.

For the two-photon jump,
\[
 \sqrt{\frac\gamma N}a_x^2
 =\sqrt{\gamma N}\,\alpha_x^2+L_x^{(1)}+N^{-1/2}L_x^{(2)}.
\]
The scalar part is removed by the standard Lindblad gauge and combines with the moving-frame Hamiltonian.  Expanding the remaining dissipator gives the exact identity
\[
 \cD[L_x^{(1)}+N^{-1/2}L_x^{(2)}]^*
 =\cD[L_x^{(1)}]^*+N^{-1/2}\mathcal B_{L_x^{(1)},L_x^{(2)}}^*
  +N^{-1}\cD[L_x^{(2)}]^*.
\]
Since $\cD[L_x^{(2)}]^*=\gamma\cD[b_x^2]^*$, summing over $x$ proves \eqref{eq:exact-shifted-generator-expansion}.

For every field monomial $M$ of degree $r$ on the fixed torus,
\begin{equation}\label{eq:normal-ordering-number-inequality}
 \|M(1+\mathfrak N_b)^{-r/2}\|\le C_{r,\Lambda}.
\end{equation}
Normal ordering the cubic cross term applied to $G$, followed by \eqref{eq:normal-ordering-number-inequality} and Cauchy--Schwarz, proves \eqref{eq:remainder-number-bounds}.  Finally $b_x^2$ lowers $\mathfrak N_b$ by two while $b_x^{*2}b_x^2$ commutes with it, so
\[
 \cD[b_x^2]^*G
 =b_x^{*2}b_x^2[(1+\mathfrak N_b-2)^2-(1+\mathfrak N_b)^2]
 =-4\mathfrak N_b b_x^{*2}b_x^2,
\]
which proves \eqref{eq:residual-two-photon-coercivity}.
\end{proof}

\begin{lemma}\label[lemma]{lem:fixed-cell-fluctuation-moments}
Fix a finite torus $\Lambda$, a bounded coherent trajectory $u(t)$ on $0\le t\le T$, and the projected coherent initial data \eqref{eq:trajectory-projected-coherent-initial-data}.  Choose $c\ge c_T$ as above.  Then
\[
 \sup_{N\ge1}\sup_{0\le t\le T}
 \widetilde\rho_{\Lambda,N}(t)((1+\mathfrak N_b)^2)
 \le C_{\Lambda,T,u},
\]
uniformly over every cutoff family with ratio $c\ge c_T$.
\end{lemma}
\begin{proof}
Apply \Cref{lem:moving-frame-remainders} to $G=(1+\mathfrak N_b)^2$.  On the fixed torus write
\[
 \cL_{u(t)}^{(2)}
 =i[H^{(2)}(t),\,\cdot\,]
 +\sum_{\varkappa=1}^{M_\Lambda}\cD[L_\varkappa^{(1)}(t)]^*,
\]
where $M_\Lambda\in\mathbb N$ is the finite number of linear jumps in this chosen representation, $H^{(2)}(t)$ is quadratic, every $L_\varkappa^{(1)}(t)$ is linear in the fluctuation fields, and all coefficients are uniformly bounded on $[0,T]$.  Normal ordering $i[H^{(2)},G]$ and each $\cD[L_\varkappa^{(1)}]^*G$ produces only finitely many monomials of degree at most four.  By \eqref{eq:normal-ordering-number-inequality}, Cauchy--Schwarz, and the form inequality
\[
 \bigl|\widetilde\rho(M+M^*)\bigr|
 \le C_M\widetilde\rho((1+\mathfrak N_b)^2)+C_M
 \qquad (\deg M\le4),
\]
we obtain
\begin{equation}\label{eq:quadratic-gaussian-fourth-moment-bound}
 \left|
 \widetilde\rho\!\left(
 \cL_{u(t)}^{(2)}(1+\mathfrak N_b)^2
 \right)\right|
 \le C\widetilde\rho((1+\mathfrak N_b)^2)+C.
\end{equation}
The boundary discrepancy is exponentially small by \Cref{lem:displaced-boundary-remainder}.  Writing $\mathfrak n_x=b_x^*b_x$, the commuting number operators satisfy on the finite-particle core
\begin{align*}
 \sum_xb_x^{*2}b_x^2
 &=\sum_x\mathfrak n_x(\mathfrak n_x-1)
 \ge |\Lambda|^{-1}\mathfrak N_b^2-\mathfrak N_b,\\
 \mathfrak N_b\sum_xb_x^{*2}b_x^2
 &\ge |\Lambda|^{-1}\mathfrak N_b^3-\mathfrak N_b^2.
\end{align*}
The residual two-photon drift controls $N^{-1}(1+\mathfrak N_b)^3$ modulo $C(1+\mathfrak N_b)^2$.  For every $\epsilon_{\rm Y}>0$, the scalar Young inequality
\[
 N^{-1/2}(1+s)^{5/2}
 \le\epsilon_{\rm Y} N^{-1}(1+s)^3+C_{\epsilon_{\rm Y}}(1+s)^2
\]
absorbs the cubic remainder into this negative drift.  By \Cref{lem:moving-frame-remainders}, this is the only quartic term: the remaining nonlinear contribution is the $N^{-1/2}$ cubic cross term, which is absorbed by the preceding Young inequality.  Consequently
\[
 \widetilde\rho_{\Lambda,N}(t)
 (\widetilde\cL_{\Lambda,N}^{\rm fl}(t)G)
 \le C\widetilde\rho_{\Lambda,N}(t)(G)+C+Ce^{-c_{\rm tail}N}.
\]
The unprojected coherent product is the moving-frame vacuum at $t=0$, and projection changes fixed fluctuation moments only by an exponentially small original-occupation tail.  Gronwall proves the bound.
\end{proof}

Let $\Lambda$ be fixed, let $u(t)$ be a bounded first-moment solution on $0\le t\le T$, and choose a cutoff ratio $c\ge c_T$ supplied by \Cref{lem:cutoff-tail-propagation}.  Start the microscopic evolution from \eqref{eq:trajectory-projected-coherent-initial-data}.  Define
\[
  \Xi_N(t)=
  \bigl(\sqrt N(Q_{x,N}-q_x(t)),\sqrt N(P_{x,N}-p_x(t))\bigr)_{x\in\Lambda}^{\mathsf T}
\]
and
\[
  (V_N(t))_{ij}
  =\frac12\rho_{\Lambda,N}(t)
  \left(\overline\Xi_{N,i}(t)\overline\Xi_{N,j}(t)
  +\overline\Xi_{N,j}(t)\overline\Xi_{N,i}(t)\right),
  \qquad
  \overline\Xi_N(t)=\Xi_N(t)-\rho_{\Lambda,N}(t)(\Xi_N(t)).
\]
Set $A(t):=D_u\mathcal F_\lambda^{\Lambda}(u(t))$.  If $\mathsf L_j^\Lambda=-\Delta_j^\Lambda$ denotes the directional graph Laplacian on the fixed torus, define
\[
  D(t)=
  \bigoplus_{x\in\Lambda}
  \left[\frac\kappa2+\gamma(q_x(t)^2+p_x(t)^2)\right]I_2
  +\sum_{j=1}^dK_j\mathsf L_j^\Lambda\otimes I_2.
\]
Let $V(t)$ solve
\begin{equation}\label{eq:finite-time-gaussian-covariance}
  \dot V(t)=A(t)V(t)+V(t)A(t)^{\mathsf T}+D(t),
  \qquad V(0)=\frac12I_{2|\Lambda|}.
\end{equation}

\begin{theorem}[Covariance convergence]\label[theorem]{thm:covariance-convergence}
Under the preceding assumptions, for every fixed matrix norm there is a constant $C_{\Lambda,T,u}$ such that, for all $c\ge c_T$,
\[
  \sup_{0\le t\le T}\norm{V_N(t)-V(t)}
  \le C_{\Lambda,T,u}N^{-1/2}
       +C_{\Lambda,T,u}e^{-c_{\rm tail}N}.
\]
The constant is independent of $c$ and $N$.
\end{theorem}
\begin{proof}
Apply \eqref{eq:exact-shifted-generator-expansion} to every centered symmetrized quadratic monomial.  The moving-frame derivative cancels all $O(\sqrt N)$ linear terms by the coherent equation.  Componentwise one obtains
\begin{equation}\label{eq:componentwise-covariance-bridge}
 \frac d{dt}(V_N)_{ij}
 =(A(t)V_N+V_NA(t)^{\mathsf T}+D(t))_{ij}
 +N^{-1/2}\varepsilon_{ij}^{(3)}(t)
 +N^{-1}\varepsilon_{ij}^{(4)}(t)
 +\varepsilon_{ij}^{\rm bdry}(t).
\end{equation}
Here $\varepsilon^{(3)}$ is a finite linear combination of centered fluctuation moments of degree at most three, and $\varepsilon^{(4)}$ one of degree at most four.  A representative component is the local quadratic observable $\mathfrak n_x=b_x^*b_x$.  From \eqref{eq:cross-dissipator},
\begin{equation}\label{eq:representative-covariance-remainders}
 \mathcal B_{L_x^{(1)},L_x^{(2)}}^*(\mathfrak n_x)
 =-3\gamma\bigl(\overline\alpha_x b_x^*b_x^2
                  +\alpha_x b_x^{*2}b_x\bigr),
 \qquad
 \gamma\cD[b_x^2]^*(\mathfrak n_x)=-2\gamma b_x^{*2}b_x^2.
\end{equation}
Its uncentered second-moment equation contains these $N^{-1/2}$ cubic and $N^{-1}$ quartic terms.  Passing to the centered covariance adds products of the corresponding linear remainders; the first-moment estimate below bounds them with those orders.  Every other symmetrized quadratic monomial is one of finitely many analogous commutator calculations.  The normal-ordering inequality \eqref{eq:normal-ordering-number-inequality} and \Cref{lem:fixed-cell-fluctuation-moments} imply
\begin{equation}\label{eq:covariance-epsilon-bounds}
 \sup_{0\le t\le T}
 (|\varepsilon_{ij}^{(3)}(t)|+|\varepsilon_{ij}^{(4)}(t)|)
 \le C_{\Lambda,T,u},
 \qquad
 |\varepsilon_{ij}^{\rm bdry}(t)|\le C_{\Lambda,T,u}e^{-c_{\rm tail}N}.
\end{equation}
The diffusion term in \eqref{eq:componentwise-covariance-bridge} is obtained directly from the linearized one-photon, two-photon, and bond jumps and is the matrix $D(t)$ stated in the theorem.

Applying the displaced-generator expansion to the linear fluctuation vector yields
\[
 \dot m_N(t)=A(t)m_N(t)+r_N^{(1)}(t),\qquad
 m_N(t)=\rho_{\Lambda,N}(t)(\Xi_N(t)),
\]
with $\sup_t\|r_N^{(1)}(t)\|\le CN^{-1/2}+Ce^{-c_{\rm tail}N}$.  Projected coherent initial data give $m_N(0)=O(e^{-c_{\rm tail}N})$, so variation of constants yields
$\sup_t\|m_N(t)\|\le CN^{-1/2}+Ce^{-c_{\rm tail}N}$.  Subtracting
$m_Nm_N^{\mathsf T}$ from the uncentered second moment changes \eqref{eq:componentwise-covariance-bridge} only by a remainder at these orders.

In matrix form,
\[
 \dot V_N=A(t)V_N+V_NA(t)^{\mathsf T}+D(t)+R_N(t),\qquad
 \sup_{0\le t\le T}\|R_N(t)\|
 \le C_{\Lambda,T,u}N^{-1/2}+C_{\Lambda,T,u}e^{-c_{\rm tail}N}.
\]
The projected coherent initial covariance differs from $I/2$ by an exponentially small cutoff tail.  Since $A(t)$ is bounded on the fixed interval, its propagator is uniformly bounded for $0\le s\le t\le T$.  Variation of constants proves the asserted covariance estimate.
\end{proof}

For $n_{\rm modes}$ bosonic modes, write
\begin{equation}\label{eq:symplectic-form}
 J_{n_{\rm modes}}:=\bigoplus_{j=1}^{n_{\rm modes}}
 \begin{pmatrix}0&1\\-1&0\end{pmatrix}.
\end{equation}
For a real two-mode covariance $W\in\mathbb R^{4\times4}$, partial transpose on the second mode is represented by
\begin{equation}\label{eq:two-mode-partial-transpose-definitions}
  \Lambda_{\rm PT}:=\operatorname{diag}(1,1,1,-1),\qquad
  W^\Gamma:=\Lambda_{\rm PT}W\Lambda_{\rm PT},\qquad
  H_{\rm PT}(W):=W^\Gamma+\frac i2J_2.
\end{equation}
The partially transposed symplectic eigenvalues are the positive moduli of the eigenvalues of $iJ_2W^\Gamma$.  Define the smaller one and the strict uncertainty-violation margin by
\begin{equation}\label{eq:pt-symplectic-eigenvalue-definition}
  \widetilde\nu_-(W):=\min\bigl\{|\lambda|:\lambda\in\operatorname{Spec}(iJ_2W^\Gamma)\bigr\},
  \qquad
  \delta_{\rm PT}(W):=-\lambda_{\min}\!\left(H_{\rm PT}(W)\right).
\end{equation}
The partially transposed covariance violates the
Robertson--Schrödinger uncertainty relation exactly when $\widetilde\nu_-(W)<1/2$, equivalently
$\delta_{\rm PT}(W)>0$. For a two-mode Gaussian state this is equivalent to NPT; for an arbitrary state the violation is a sufficient NPT witness.

A matrix $\mathsf T\in\mathbb R^{4\times2n_{\rm modes}}$ is called a passive canonical two-mode selector when
\begin{equation}\label{eq:passive-canonical-selector}
  \mathsf T\mathsf T^{\mathsf T}=I_4,
  \qquad
  \mathsf T J_{n_{\rm modes}}\mathsf T^{\mathsf T}=J_2,
\end{equation}
and there exists an orthogonal symplectic matrix
$\mathsf S\in O(2n_{\rm modes})\cap\operatorname{Sp}(2n_{\rm modes},\mathbb R)$ whose first four rows are $\mathsf T$.  This extension condition identifies the selector with a passive mode unitary followed by tracing out the unselected modes.

\begin{corollary}\label[corollary]{cor:trajectory-mode-npt}
In the setting of \Cref{thm:covariance-convergence}, first choose any cutoff ratio $c\ge c_T$.  Let $t_*\in(0,T]$ and let $\mathsf T\in\mathbb R^{4\times2|\Lambda|}$ be a passive canonical two-mode selector in the preceding sense.  If
\[
  \widetilde\nu_-\!\left(\mathsf T V(t_*)\mathsf T^{\mathsf T}\right)<\frac12,
\]
then, for the chosen cutoff ratio, there exists $N_0<\infty$ such that every $N\ge N_0$ has an NPT selected two-mode reduction at time $t_*$.
\end{corollary}
\begin{proof}
Zero-extend the locally truncated state, apply a passive mode unitary implementing a symplectic extension of $\mathsf T$, and trace out the unselected modes.  Choose $0<\delta<\delta_{\rm PT}(\mathsf T V(t_*)\mathsf T^{\mathsf T})$.  After the cutoff ratio has been fixed, the covariance convergence theorem transfers that violation to the microscopic covariance for all sufficiently large $N$.  Every PPT two-mode state obeys the Robertson--Schr\"odinger uncertainty relation after partial transpose, so the selected two-mode reduced state is NPT.
\end{proof}

\begin{corollary}\label[corollary]{cor:finiteN-mode-npt}
Let $\Lambda=\mathcal C$, let $u(t)=u^\star$ be a stationary periodic branch, and assume that $A_{\rm per}$ is Hurwitz.  Write $V_{N,c}(t)$ for the microscopic covariance with cutoff $M_N=\lfloor cN\rfloor$.  For every tolerance $\epsilon_{\rm tol}>0$, there are $t_{\epsilon_{\rm tol}}<\infty$ and $c_{\epsilon_{\rm tol}}<\infty$ such that, for every $c\ge c_{\epsilon_{\rm tol}}$,
\[
  \limsup_{N\to\infty}
  \norm{V_{N,c}(t_{\epsilon_{\rm tol}})-V_{\rm per}}
  \le\epsilon_{\rm tol}.
\]
If a passive canonical two-mode selector $\mathsf T$ satisfies
\[
  \widetilde\nu_-\!\left(\mathsf T V_{\rm per}\mathsf T^{\mathsf T}\right)<\frac12,
\]
then there are finite $t_0,c_0,N_0$ such that, for every $c\ge c_0$ and $N\ge N_0$, the selected two-mode reduced state of the canonically embedded microscopic state at time $t_0$ is NPT.
\end{corollary}
\begin{proof}
For a stationary branch, \eqref{eq:finite-time-gaussian-covariance} reduces to \eqref{eq:period-cell-lyapunov}, whose solution converges exponentially to $V_{\rm per}$.  Given $\epsilon_{\rm tol}>0$, first choose $t_{\epsilon_{\rm tol}}$ so that $\norm{V(t_{\epsilon_{\rm tol}})-V_{\rm per}}<\epsilon_{\rm tol}/2$.  Next choose $c_{\epsilon_{\rm tol}}\ge c_{t_{\epsilon_{\rm tol}}}$ from \Cref{thm:covariance-convergence}.  The covariance-convergence estimate and the triangle inequality give the displayed bound.  If $\mathsf T V_{\rm per}\mathsf T^{\mathsf T}$ has a strict NPT margin, continuity of the partially transposed symplectic spectrum gives a sufficiently late finite $t_0$ for the limiting Gaussian covariance; choosing the finite relaxation time $t_0$, then $c_0\ge c_{t_0}$, and finally $N_0$ allows \Cref{cor:trajectory-mode-npt} to transfer that strict violation to the microscopic two-mode reduction.
\end{proof}

For a one-dimensional period cell, set
$\mathcal C^*:=(2\pi/|\mathcal C|)(\mathbb Z/|\mathcal C|\mathbb Z)$.  Choose $k\in\mathcal C^*$ with $k\ne0$ and $k\not\equiv-k\pmod{2\pi}$, thereby excluding the self-conjugate modes $0$ and, when present, $\pi$.  Define $b_k:=|\mathcal C|^{-1/2}\sum_{x\in\mathcal C}e^{-ikx}b_x$.  The corresponding real quadratures are
\begin{equation}\label{eq:real-fourier-symplectic}
\begin{pmatrix}Q_k\\ P_k\\ Q_{-k}\\ P_{-k}\end{pmatrix}
=\frac1{\sqrt{|\mathcal C|}}\sum_{x\in\mathcal C}
\begin{pmatrix}
 \cos(kx)& \sin(kx)\\
 -\sin(kx)& \cos(kx)\\
 \cos(kx)&-\sin(kx)\\
 \sin(kx)& \cos(kx)
\end{pmatrix}
\binom{\delta q_x}{\delta p_x}
=: \mathsf T_{k,-k}\xi.
\end{equation}
The Fourier rows obey
$\mathsf T_{k,-k}\mathsf T_{k,-k}^{\mathsf T}=I_4$ and
$\mathsf T_{k,-k}J_{|\mathcal C|}\mathsf T_{k,-k}^{\mathsf T}=J_2$; completing the real Fourier transform gives the required orthogonal symplectic extension.  Thus $\mathsf T_{k,-k}$ is a passive canonical selector, and
\[
  V_{k,-k}:=\mathsf T_{k,-k}V_{\rm per}\mathsf T_{k,-k}^{\mathsf T},
  \qquad
  \widetilde\nu_-(k,-k):=\widetilde\nu_-(V_{k,-k}).
\]
By the two-mode Gaussian PPT criterion \cite{Simon2000,Duan2000,VidalWerner2002,Weedbrook2012}, the pair is NPT exactly when $\widetilde\nu_-(k,-k)<1/2$, with logarithmic negativity
\[
  E_{\rm LN}(k,-k):=\max\{0,-\log_2[2\widetilde\nu_-(k,-k)]\}.
\]
Define the standing-wave quadratures by
\[
 Q_+:=\frac{Q_k+Q_{-k}}{\sqrt2},\qquad
 P_+:=\frac{P_k+P_{-k}}{\sqrt2},\qquad
 Q_-:=\frac{Q_k-Q_{-k}}{\sqrt2},\qquad
 P_-:=\frac{P_k-P_{-k}}{\sqrt2}.
\]
In the standing-wave order $(Q_+,P_+,Q_-,P_-)$ associated with $A_k^{(+)}\oplus A_k^{(-)}$ and the traveling-mode order $(Q_k,P_k,Q_{-k},P_{-k})$, the passive transformation is
\begin{equation}\label{eq:standing-to-traveling-symplectic}
  \mathsf S_{\rm sw\to tr}:=\frac1{\sqrt2}
  \begin{pmatrix}
    1&0&1&0\\
    0&1&0&1\\
    1&0&-1&0\\
    0&1&0&-1
  \end{pmatrix},
  \qquad
  \mathsf S_{\rm sw\to tr}\in O(4)\cap\operatorname{Sp}(4,\mathbb R).
\end{equation}

\begin{proposition}\label[proposition]{prop:homogeneous-covariance}
At the homogeneous state $u^\star=0$, set
\[
  \omega_j(k)=2-2\cos k_j,
  \qquad
  \kappa_k=\kappa+2\sum_{j=1}^dK_j\omega_j(k),
  \qquad
  g_k=\varepsilon-\sum_{j=1}^dK_j'\omega_j(k),
\]
\[
  R_k=\sqrt{\kappa_k^2+4\Omega_\lambda^2},
  \qquad
  \eta_k=\frac{2|g_k|}{R_k},
  \qquad
  \Delta_k=R_k^2-4g_k^2,
  \qquad
  s_k=\frac{R_k^2}{2\Delta_k}.
\]
For a non-self-conjugate pair $(k,-k)$, a real standing-wave basis gives
\begin{equation}\label{eq:homogeneous-standing-wave-dynamics}
  A_{k,-k}^{\rm sw}=A_k^{(+)}\oplus A_k^{(-)},
  \qquad
  D_{k,-k}^{\rm sw}=\frac{\kappa_k}{2}I_4,
\end{equation}
where
\begin{equation}\label{eq:homogeneous-standing-wave-blocks}
  A_k^{(\pm)}=
  \begin{pmatrix}
   -\kappa_k/2\pm g_k & \Omega_\lambda\\
   -\Omega_\lambda & -\kappa_k/2\mp g_k
  \end{pmatrix},
  \qquad
  D_k^{(\pm)}=\frac{\kappa_k}{2}I_2.
\end{equation}
Here $A_k^{(+)}$ is the classical Fourier symbol in the chosen real quadratures.  If $\eta_k<1$, the stationary covariance of the traveling modes $(b_k,b_{-k})$, in the order $(Q_k,P_k,Q_{-k},P_{-k})$, is
\begin{equation}\label{eq:homogeneous-traveling-covariance}
  V_{k,-k}^{\rm hom}
  =\begin{pmatrix}
      s_kI_2&C_k\\
      C_k&s_kI_2
    \end{pmatrix},
  \qquad
  C_k=\frac{g_k}{\Delta_k}
  \begin{pmatrix}
    \kappa_k&-2\Omega_\lambda\\
    -2\Omega_\lambda&-\kappa_k
  \end{pmatrix}.
\end{equation}
Both ordinary symplectic eigenvalues are
\begin{equation}\label{eq:homogeneous-spectrum}
  \nu_{\rm phys}(k)
  =\frac{R_k}{2\sqrt{\Delta_k}}
  =\frac{1}{2\sqrt{1-\eta_k^2}}.
\end{equation}
\end{proposition}

\begin{theorem}[Homogeneous stability and opposite-momentum entanglement]\label[theorem]{thm:homogeneous-spectrum}
Let $(k,-k)$ be a non-self-conjugate homogeneous momentum pair and use the quantities in \Cref{prop:homogeneous-covariance}.  The pair is stable if and only if $\eta_k<1$.  In the stable regime,
\begin{equation}\label{eq:homogeneous-pt-spectrum}
  \widetilde\nu_-(k,-k)
  =\frac{R_k}{2(R_k+2|g_k|)},
  \qquad
  \widetilde\nu_+(k,-k)
  =\frac{R_k}{2(R_k-2|g_k|)},
\end{equation}
and hence
\begin{equation}\label{eq:homogeneous-npt}
  \widetilde\nu_-(k,-k)=\frac{1}{2(1+\eta_k)},
  \qquad
  E_{\rm LN}(k,-k)=\log_2(1+\eta_k).
\end{equation}
The classical Fourier block satisfies
\[
  \det A_k^{(+)}=\frac{R_k^2}{4}(1-\eta_k^2).
\]
\end{theorem}
\begin{proof}[Proof of \Cref{prop:homogeneous-covariance,thm:homogeneous-spectrum}]
The homogeneous Fourier transform decomposes the pair $(k,-k)$ into two orthogonal standing-wave modes with equal damping $\kappa_k$ and opposite single-mode squeezing parameters $\pm g_k$, as in \eqref{eq:homogeneous-standing-wave-dynamics}--\eqref{eq:homogeneous-standing-wave-blocks}.  Solving the two $2\times2$ Lyapunov equations yields
\[
  V_k^{(+)}=
  \frac{1}{\Delta_k}
  \begin{pmatrix}
    R_k^2/2+\kappa_kg_k & -2\Omega_\lambda g_k\\
    -2\Omega_\lambda g_k & R_k^2/2-\kappa_kg_k
  \end{pmatrix},
  \qquad
  V_k^{(-)}=V_k^{(+)}\big|_{g_k\mapsto-g_k}.
\]
Both drift blocks have trace $-\kappa_k<0$ and determinant $\Delta_k/4$, so their Hurwitz condition is $\eta_k<1$.  The explicit passive symplectic matrix $\mathsf S_{\rm sw\to tr}$ in \eqref{eq:standing-to-traveling-symplectic} recombines the two standing waves into the traveling modes $(b_k,b_{-k})$ and sends $V_k^{(+)}\oplus V_k^{(-)}$ to \eqref{eq:homogeneous-traveling-covariance}.  A direct symplectic calculation then gives the repeated ordinary eigenvalue \eqref{eq:homogeneous-spectrum}.  Partial transpose changes the sign of $P_{-k}$ and gives \eqref{eq:homogeneous-pt-spectrum}; the smaller eigenvalue and the logarithmic negativity are then \eqref{eq:homogeneous-npt}.
\end{proof}

\begin{corollary}\label[corollary]{cor:common-temperature}
Suppose every linear damping channel contributing to a stable homogeneous pair is coupled to a common thermal occupation $\bar n\ge0$.  The drift is unchanged and the diffusion is multiplied by $2\bar n+1$.  Then
\begin{equation}\label{eq:thermal-homogeneous-npt}
 \widetilde\nu_-^{(\bar n)}(k,-k)
 =\frac{2\bar n+1}{2(1+\eta_k)},\qquad
 E_{\rm LN}^{(\bar n)}(k,-k)
 =\max\left\{0,\log_2\frac{1+\eta_k}{2\bar n+1}\right\}.
\end{equation}
In particular
\[
 \text{NPT}\quad\Longleftrightarrow\quad \eta_k>2\bar n.
\]
\end{corollary}
\begin{proof}
The Lyapunov equation is linear in the diffusion matrix, so the stationary covariance in \eqref{eq:homogeneous-traveling-covariance} is multiplied by $2\bar n+1$.  Symplectic eigenvalues scale by this factor, and \eqref{eq:thermal-homogeneous-npt} follows from \eqref{eq:homogeneous-npt}.
\end{proof}

The bipartition in \Cref{thm:homogeneous-spectrum} is the physically distinguished pair of counterpropagating Fourier modes $b_k$ and $b_{-k}$.  The parametric term squeezes the two standing-wave modes in orthogonal quadratures; their passive recombination into traveling modes therefore produces two-mode NPT entanglement whenever $g_k\ne0$.

\begin{corollary}\label[corollary]{cor:homogeneous-selection}
In the homogeneous Gaussian sector, every stable non-self-conjugate pair with $g_k\ne0$ is NPT.  Within the stable homogeneous phase, minimizing $\widetilde\nu_-(k,-k)$, maximizing $E_{\rm LN}(k,-k)$, and maximizing $\eta_k$ are equivalent spectral optimizations.  The classical determinant is
\begin{equation}\label{eq:turing-entanglement-ratio}
  \det A_k^{(+)}=\frac{R_k^2}{4}(1-\eta_k^2),
\end{equation}
so the Turing threshold at $k_*$ is the condition $\eta_{k_*}=1$.  Consequently
\[
  \lambda\uparrow0:
  \qquad
  \nu_{\rm phys}(k_*)\longrightarrow\infty,
  \qquad
  \widetilde\nu_-(k_*,-k_*)\longrightarrow\frac14,
  \qquad
  E_{\rm LN}(k_*,-k_*)\longrightarrow1.
\]
No finite stationary covariance exists at threshold because the critical drift mode is neutral, but the two NPT quantities have the displayed one-sided limits.
\end{corollary}
\begin{proof}
Equations \eqref{eq:homogeneous-spectrum}--\eqref{eq:homogeneous-npt} are monotone functions of $\eta_k$, while \eqref{eq:turing-entanglement-ratio} follows from the determinant in \eqref{eq:homogeneous-standing-wave-blocks}.
\end{proof}

\begin{corollary}\label[corollary]{cor:transport-selection}
If the transport is scalar, then $K_j'=0$ for every $j$ and $g_k=\varepsilon$ is momentum independent.  Hence $\eta_k$ decreases with the weighted transport symbol
\[
  \mathfrak t(k):=\sum_{j=1}^dK_j\omega_j(k),
\]
and the strongest opposite-momentum NPT occurs at nonzero modes minimizing $\mathfrak t(k)$.  In the isotropic scalar case these are the smallest nonzero momenta.  For the explicit differential-transport stripe family on the stable side, the longitudinal maximum of $\eta_k$ occurs at
\begin{equation}\label{eq:homogeneous-entanglement-wave-number}
  \omega_{\rm ent}(\lambda)
  =2-\sqrt3-\left(\frac{2}{\sqrt3}-1\right)\lambda,
  \qquad \lambda<0,
\end{equation}
for $\lambda$ sufficiently close to zero.  The homogeneous NPT minimum is at a finite wave number and converges to the Turing value $\omega_*=2-\sqrt3$ as $\lambda\uparrow0$.
\end{corollary}
\begin{proof}
Under scalar transport, $R_k$ increases monotonically with $\mathfrak t(k)$ while $|g_k|$ is fixed.  For the explicit stripe parameters, differentiation of $\eta_k^2$ with respect to the longitudinal lattice symbol $\omega$ gives a unique interior maximum at \eqref{eq:homogeneous-entanglement-wave-number}.
\end{proof}

\begin{figure}[!t]
\centering
\includegraphics[width=0.99\textwidth]{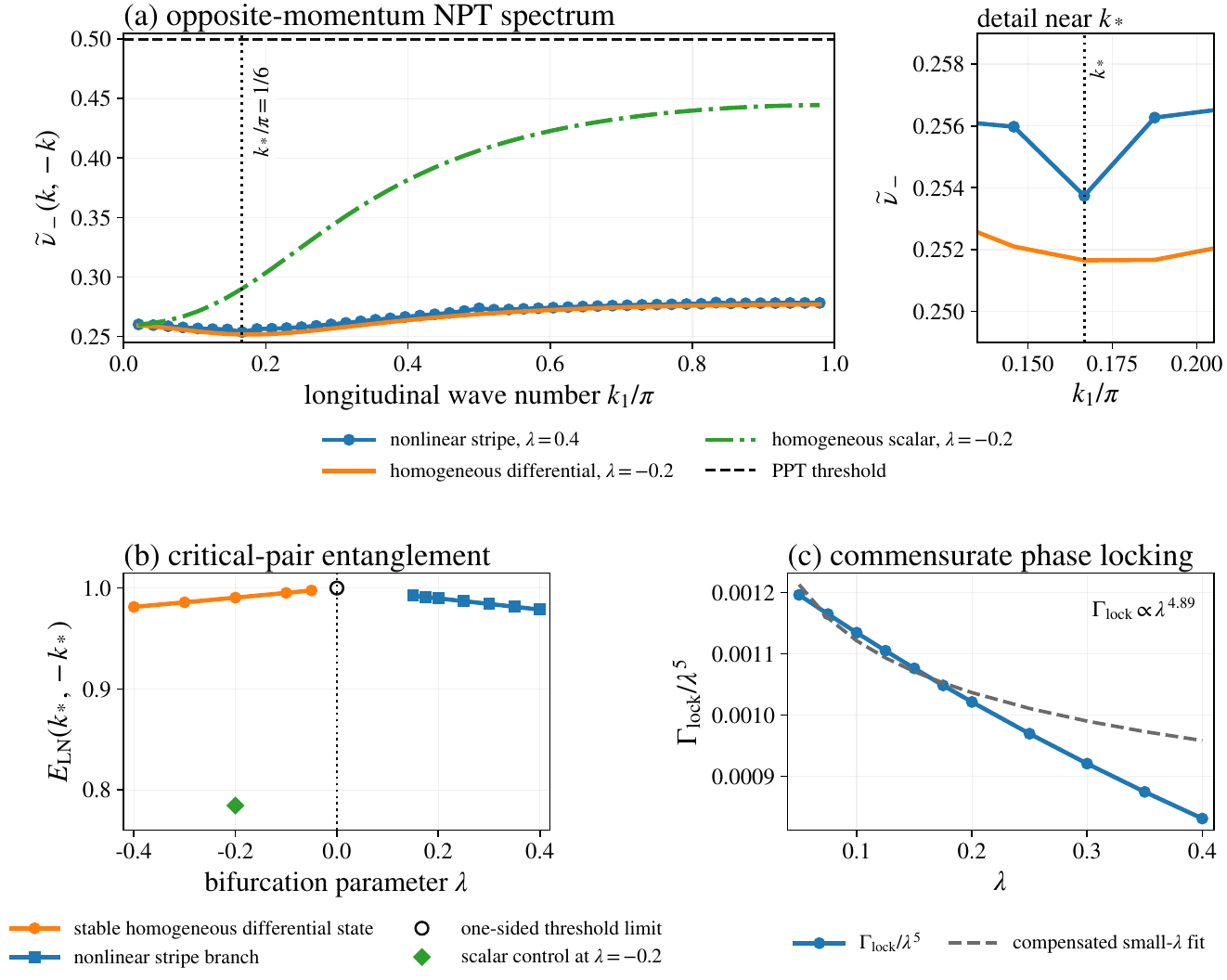}
\caption[Momentum-resolved Gaussian entanglement]{Momentum-resolved Gaussian entanglement and phase locking.  (a) Longitudinal opposite-momentum spectra for the nonlinear stripe, the homogeneous differential-transport system, and a scalar-transport control; the detail panel resolves the neighborhood of $k_*$.  (b) Critical-pair logarithmic negativity across the homogeneous threshold and along the continued patterned branch.  (c) Phase-locking rate after division by the symmetry-predicted $\lambda^5$ scale.}
\label{fig:mode-entanglement}
\end{figure}
\FloatBarrier

Starting from the bifurcating branch, Newton continuation reaches the representative nonlinear value $(\lambda,\nu)=(0.4,4)$.  We apply \Cref{prop:branch-gaussian} to the unseeded generator ($h=0$) and the unseeded bond-centered stationary stripe.  At this value the transverse-uniform drift is Hurwitz with
\[
  \max\{\operatorname{Re}z:z\in\operatorname{Spec}(A_{\rm per})\}=-8.51\times10^{-6}.
\]
Let $V_{\rm per}$ denote the exact solution of
$A_{\rm per}V+VA_{\rm per}^{\mathsf T}+D_{\rm per}=0$, and let
$\widehat V_{\rm per}$ denote the floating-point Lyapunov solution used in the computation.  For a real $n\times n$ matrix $A$, define the Lyapunov operator and its two-norm separation by
\begin{equation}\label{eq:lyapunov-separation-definition}
  \mathcal K_A(X):=AX+XA^{\mathsf T},\qquad
  \operatorname{sep}_2(\mathcal K_A):=
  \sigma_{\min}(I_n\otimes A+A\otimes I_n),
\end{equation}
where $\sigma_{\min}$ denotes the smallest singular value.  The floating-point solution has relative residual
\[
  r_{\rm Lyap}:=
  \frac{\|A_{\rm per}\widehat V_{\rm per}
  +\widehat V_{\rm per}A_{\rm per}^{\mathsf T}+D_{\rm per}\|_{\rm F}}
       {\|D_{\rm per}\|_{\rm F}}
  =7.36\times10^{-11}.
\]
The vectorized Lyapunov operator has two-norm separation
$\operatorname{sep}_2(\mathcal K_{A_{\rm per}})=9.92\times10^{-6}$.
The absolute Frobenius residual $3.64\times10^{-9}$ and the separation estimate yield the a posteriori error bound
\begin{equation}\label{eq:period-cell-covariance-error-bound}
 \|V_{\rm per}-\widehat V_{\rm per}\|_{\rm F}
 \le 3.66\times10^{-4}.
\end{equation}
The numerical covariance is physical, with minimum symplectic eigenvalue $0.823$.  The critical pair gives
\[
  \widetilde\nu_-(k_*,-k_*)=0.254,
  \qquad
  E_{\rm LN}(k_*,-k_*)=0.979.
\]
Let
\[
  W_*^{\rm exact}:=\mathsf T_{k_*,-k_*}V_{\rm per}\mathsf T_{k_*,-k_*}^{\mathsf T},\qquad
  \widehat W_*:=\mathsf T_{k_*,-k_*}\widehat V_{\rm per}\mathsf T_{k_*,-k_*}^{\mathsf T},
\]
and define $H_{\rm PT}^{\rm exact}:=H_{\rm PT}(W_*^{\rm exact})$ and
$\widehat H_{\rm PT}:=H_{\rm PT}(\widehat W_*)$.  Because the selector has orthonormal rows and $\Lambda_{\rm PT}$ is orthogonal,
\[
 \|W_*^{\rm exact}-\widehat W_*\|_2
 \le\|V_{\rm per}-\widehat V_{\rm per}\|_2
 \le\|V_{\rm per}-\widehat V_{\rm per}\|_{\rm F},
 \qquad
 \|H_{\rm PT}^{\rm exact}-\widehat H_{\rm PT}\|_2
 =\|W_*^{\rm exact}-\widehat W_*\|_2.
\]
The numerical value is $\lambda_{\min}(\widehat H_{\rm PT})=-0.246$ using unrounded data, so Weyl's inequality and \eqref{eq:period-cell-covariance-error-bound} give
\[
 \lambda_{\min}(H_{\rm PT}^{\rm exact})
 \le -0.246+3.66\times10^{-4}<-0.245<0.
\]  An independent high-precision PPT invariant in the reproducibility archive confirms the sign.  \Cref{cor:finiteN-mode-npt} transfers this strict violation to sufficiently large microscopic systems initialized on the projected coherent stripe.

For the explicit stripe at $\lambda=0.4$, this pair attains the global minimum among the resolved longitudinal opposite-momentum pairs $k=(k_1,0)$ with $0<k_1<\pi$ in the transverse-uniform $L=48$ and $L=96$ supercells.  The minimum among nearest-neighbor real-space two-mode reductions is instead $\widetilde\nu_-=0.870>1/2$.

The closed formulas in \Cref{thm:homogeneous-spectrum,cor:homogeneous-selection,cor:transport-selection} explain the two homogeneous controls at the chosen parameter value.  At $\lambda=-0.2$, write $\omega_{\rm ent}:=\omega_{\rm ent}(-0.2)$ and let $k_{\rm ent}\in(0,\pi)$ be the unique longitudinal wave number determined by
\[
  2(1-\cos k_{\rm ent})=\omega_{\rm ent}.
\]
The continuous differential-transport spectrum is minimized at
\[
  \omega_{\rm ent}=0.299,
  \qquad
  k_{\rm ent}/\pi=0.176,
\]
which lies close to $k_*/\pi=1/6$; on the $L=96$ longitudinal momentum grid the minimum occurs at $k_*$ and equals $\widetilde\nu_-=0.252$.  Keeping $\Omega_\lambda$ fixed and using scalar transport $D_q=D_p=2+\sqrt3$, the minimum moves to the smallest nonzero longitudinal momentum, $k_1/\pi=1/48$, where $\widetilde\nu_-=0.260$; at $k_*$ the value is $0.290$.  Differential transport therefore relocates the homogeneous opposite-momentum NPT minimum from the long-wavelength edge to a finite wave number, while the representative nonlinear stripe retains the commensurate minimum in the transverse-uniform supercells examined here.

The leading drift eigenvector has overlap $0.997$ with the phase tangent $-r\sin(k_*x+\pi/12)$, identifying the soft spectral abscissa as the commensurate phase-locking rate.  Its observed scaling is consistent with the first symmetry-allowed $C_{12}$ locking term.

\subsection{Finite-time fluctuation spectra for spot and labyrinth}
To obtain the fluctuation spectra for the spot and labyrinth regimes, we propagate the nonautonomous Lyapunov equation along each first-moment trajectory $u(t)$:
\[
 \dot V(t)=A(t)V(t)+V(t)A(t)^{\mathsf T}+D(t),
 \qquad
 A(t):=D_u\mathcal F_\lambda^{\Lambda_L}(u(t)),
 \qquad
 V(0)=\frac12I_{2L^2},
\]
where $A(t),D(t),V(t)\in\mathbb R^{2L^2\times2L^2}$ and $D(t)$ is assembled from the one-photon, two-photon, and dissipative-bond jumps as in \eqref{eq:period-cell-lyapunov}.  Let $U(t,s)\in\mathbb R^{2L^2\times2L^2}$ be the fundamental matrix defined by
\[
  \partial_tU(t,s)=A(t)U(t,s),
  \qquad U(s,s)=I_{2L^2}.
\]
Only selected $4\times4$ Fourier-pair reductions are required.  Let $\mathsf T_{k,-k}\in\mathbb R^{4\times2L^2}$ be the passive selector for $(k,-k)$ and define
$Z_k(\tau):=U(T,T-\tau)^{\mathsf T}\mathsf T_{k,-k}^{\mathsf T}\in\mathbb R^{2L^2\times4}$.  Variation of constants yields the low-rank identity
\[
 \mathsf T_{k,-k}V(T)\mathsf T_{k,-k}^{\mathsf T}
 =\frac12 Z_k(T)^{\mathsf T}Z_k(T)
 +\int_0^T Z_k(\tau)^{\mathsf T}D(T-\tau)Z_k(\tau)\,d\tau.
\]
We integrate the adjoint equation by Strang splitting: the translation-invariant linear part is exponentiated analytically in Fourier space and the local nonlinear Jacobian is exponentiated pointwise.  The noise integral uses the trapezoidal rule.  This avoids constructing the full $2L^2\times2L^2$ covariance.

At the finest covariance step, all seven sampled covariances are physical before partial transpose, and each selected-shell direction is NPT.  The shell remains separated from both radial controls under covariance-step refinement and coherent-step halving.  The minimum shell/control gaps are $0.010$ for Spot and $0.005$ for Labyrinth; the largest changes caused by halving $\Delta t_{\rm PDE}$ are $2.44\times10^{-4}$ and $1.25\times10^{-3}$, respectively.  Direction-resolved values and Richardson extrapolations are included in the reproducibility archive.  \Cref{cor:trajectory-mode-npt} transfers each strict finite-time violation to sufficiently large microscopic systems.

\begin{figure}[!htbp]
\centering
\includegraphics[width=0.82\textwidth]{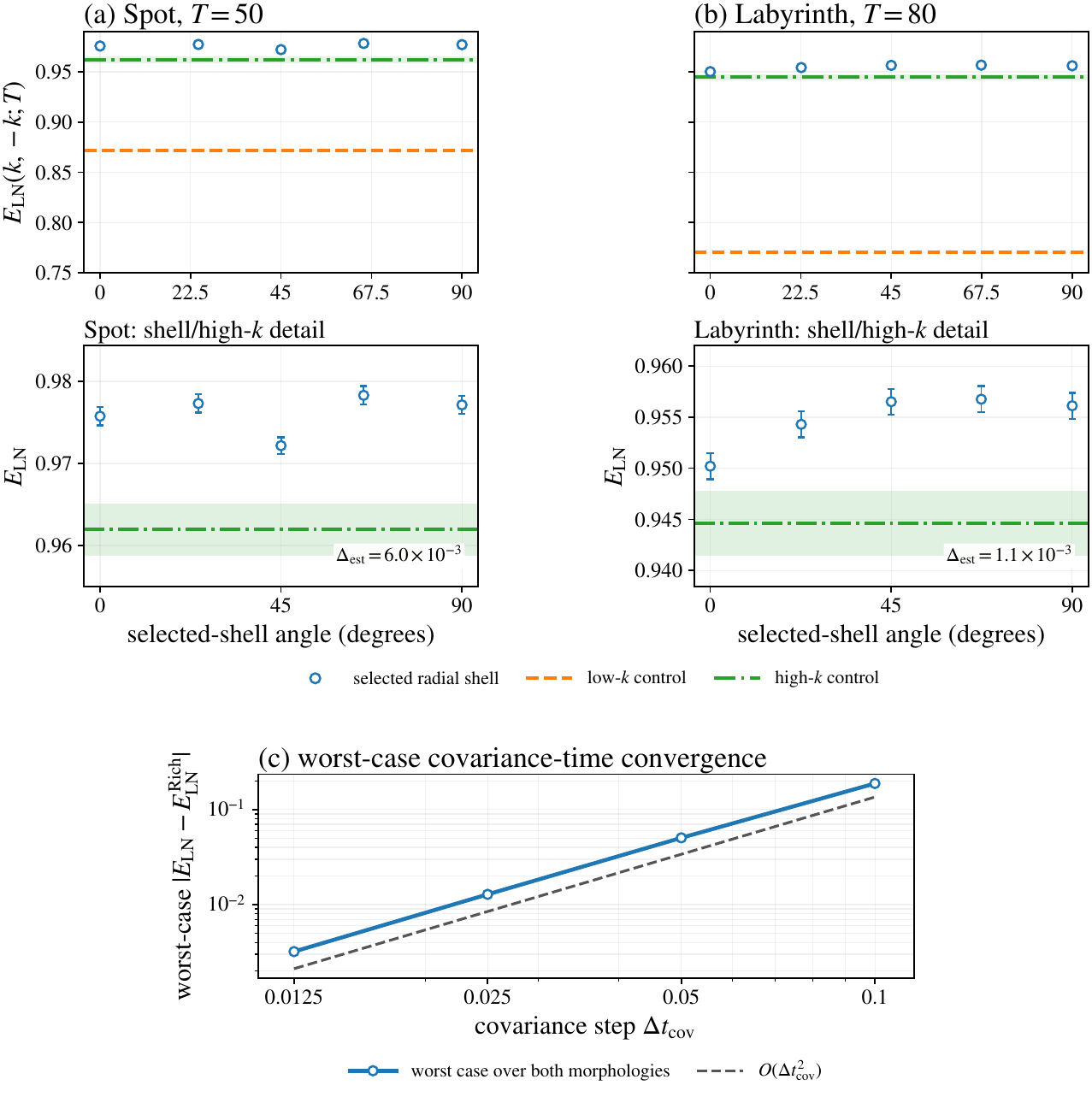}
\caption[Opposite-momentum entanglement along finite-time trajectories]{Opposite-momentum entanglement along the finite-time Spot and Labyrinth trajectories.  Panels (a) and (b) compare five directions on the selected radial shell with low- and high-wave-number controls.  Panel (c) shows the covariance-step differences for the sampled modes.}
\label{fig:isotropic-fluctuations}
\end{figure}
\FloatBarrier

\clearpage
\section{Conclusion}
Completely positive lattice dynamics can select a nonzero Turing scale and sustain quantum order at the microscopic level.  Our explicit family combines a supercritical bifurcation, stable commensurate stripe branches, and extensive Bragg order with a controlled fluctuation limit.  Microscopic covariances converge at rate $O(N^{-1/2})$ to the Gaussian Lyapunov flow, carrying strict partial-transpose uncertainty violations to large $N$.

In the homogeneous Gaussian sector, a single spectral ratio governs both classical stability and opposite-momentum entanglement.  The two-dimensional stripe, spot, and labyrinth computations show the same selected scale organizing morphology and quantum correlations.  Together, the construction and fluctuation theory establish Lindblad dynamics as a microscopic origin of quantum Turing order.

\section*{Acknowledgments}
This work was supported in part by the U.S. National Science Foundation under Grant No.~OSI-2328774.



\section*{Data and code availability.}
The numerical source, reference data, and Lean~4 formalization are provided in the author's GitHub repository \url{https://github.com/IKEDAKAZUKI/Quantum-Turing-Pattern}.

\renewcommand{\bibfont}{\footnotesize}
\setlength{\bibsep}{0pt}
\bibliographystyle{unsrtnat}
\bibliography{references}

\end{document}